\documentclass[11pt]{article}

\usepackage{scrpage2} 

\usepackage[a4paper, left=2.2cm, right=2.2cm, bottom=2.8cm, top=2.0cm,footskip=1.3cm]{geometry}
\usepackage{bm,epsfig,booktabs} 
\usepackage{setspace}
\usepackage{colortbl}
\usepackage{caption}
\usepackage{subcaption}
\usepackage{relsize} 
\usepackage{graphicx}
\usepackage{cancel}   
\usepackage{slashed}
\usepackage{verbatim}   
\usepackage{mathtools}
\usepackage[flushmargin,hang]{footmisc}
\usepackage{amsmath,enumerate} 
\usepackage[ntheorem]{defoqua} 
\usepackage{stmaryrd}   
\usepackage{cite}
\usepackage[arrow, matrix,curve]{xy}
\usepackage{array}
\usepackage{color}
\usepackage[usenames,dvipsnames]{xcolor}
\input xypic
\usepackage{hyperref}

\newcommand{\hs} {\hspace{1pt}} 
\newcommand{\itspace} {\vspace{-1.2ex}}
\newcommand{\itspacec} {\vspace{-1.2ex}} 
\newcommand{\itspacecc} {\vspace{-0.5ex}}

\newcommand{\w} {\omega}
\newcommand{\Con} {\mathcal{A}}
\newcommand{\A} {\ovl{\Con}}
\newcommand{\AQR} {\ovl{\Con}_\inv}
\newcommand{\ARQ} {\ovl{\Con_\inv}}
\newcommand{\ConG} {\Con_\inv} 
\newcommand{\ishomc} {\mathfrak{v}}
\newcommand{\oishomc} {\ovl{\mathfrak{v}}}

\newcommand{\Borel}{\mathfrak{B}}
\newcommand{\BRq}{\Borel\big(\hspace{1pt}\qR\hspace{1pt}\big)}
\newcommand{\mAL}{\mu_{\mathrm{AL}}}
\newcommand{\mus}{\mu_{\mathrm{s}}}
\newcommand{\muL}{\lambda}
\newcommand{\muB} {\mu_{\mathrm{Bohr}}} 

\newcommand{\RB} {{\mathbb{R}_{\mathrm{Bohr}}}}
\newcommand{\CCC} {\mathbb{C}}
\newcommand{\RR} {\mathbb{R}}
\newcommand{\NN} {\mathbb{N}}
\newcommand{\ZZ} {\mathbb{Z}}
\newcommand{\ZN}{\mathbb{Z}\backslash\{0\}}
\newcommand{\Q}{\mathbb{Q}}

\newcommand{\murs} {\mathfrak{z}} 
\newcommand{\mj} {\mathfrak{j}}
\newcommand{\homiso} {\kappa}
\newcommand{\Add}[1] {\Ad_{#1}}

\newcommand{\IR} {\mathcal{R}}
\newcommand{\F} {\mathrm{F}}
\newcommand{\HHH} {\mathrm{H}}
\newcommand{\me} {\mathbb{1}}
\newcommand{\TT} {\mathfrak{T}}
\newcommand{\TO} {\mathrm{T}}
\newcommand{\DDD} {\mathrm{D}}
\newcommand{\res} {\mathfrak{r}}
\newcommand{\III} {\mathcal{I}}
\newcommand{\dd} {\mathrm{d}}
\newcommand{\inv} {\mathrm{inv}}
\newcommand{\gr} {\Gamma_0}
\newcommand{\aA} {\mathfrak{A}}
\newcommand{\cC} {\mathfrak{C}}
\newcommand{\gG} {\mathfrak{G}}
\newcommand{\Ge}{G}
\newcommand{\Gee} {\mathbb{R}^3 \rtimes_\uberl \SU}
\newcommand{\AP}{\mathrm{AP}}
\newcommand{\Ts}{\TO_{\mathrm{s}}}
\newcommand{\OO}{\mathcal{O}}
\newcommand{\chih}{\widehat{\chi}}
\newcommand{\pih}{\widehat{\pi}}
\newcommand{\wq}{\ovl{\omega}}
\newcommand{\RRD} {\widehat{\RR}}
\newcommand{\LGA} {\Theta}
\newcommand{\INDA} {\vartheta}
\newcommand{\conac} {\phi}
\newcommand{\mmu} {\wt{\mu}}
\newcommand{\Asp} {\Phi}
\newcommand{\SU} {\mathrm{SU}(2)}
\newcommand{\su} {\mathrm{su}(2)}
\newcommand{\SO} {\mathrm{SO}(3)}
\newcommand{\wt}[1] {\widetilde{#1}} 
\newcommand{\ovl}[1]{\overline{#1}} 
\newcommand{\cp} {\circ}
\newcommand{\im}{\mathrm{im}} 
\newcommand{\dom} {\mathrm{dom}}
\newcommand{\Spec}{\mathrm{Spec}}
\newcommand{\pred}{\Paths_{\mathrm{red}}}
\newcommand{\gred}{\Gamma_0}
\newcommand{\Hil} {\mathcal{H}}
\newcommand{\Lzw}[2] {L^{2}\left(#1,#2\right)}
\newcommand{\Rq}{\raisebox{0.2ex}{$($}\raisebox{-0.1ex}{$\qR$}\raisebox{0.2ex}{$)$}}
\newcommand{\pillstr} {\big(\raisebox{-0.1ex}{$\pi^{L'}_{L}$}\big)}
\newcommand{\BTK} {\raisebox{0pt}{$\Borel\raisebox{1pt}{$\big($}\raisebox{-0ex}{$\TO^{|L|}$}\raisebox{1pt}{$\big)$}$}}
\newcommand{\B} {\mathrm{B}}
\newcommand{\Stab}[2] {\mathrm{Stab}_{#1}(#2)}
\newcommand{\hommm} {\epsilon}
\newcommand{\hx} {\hat{x}}
\newcommand{\Homm} {\mathrm{Hom}} 
\newcommand{\T} {\mathcal{T}} 
\newcommand{\Paths} {\mathcal{P}}
\newcommand{\mm} {\tau}
\newcommand{\parall}[2] {\mathcal{P}_{#1}^{#2}}
\newcommand{\uberl} {\varrho}
\newcommand{\uberll}[2] {#1(#2)}
\newcommand{\Co}[1] {\alpha_{#1}}
\newcommand{\gc}[4] {\gamma_{\vec{#1},\vec{#2}}^{#3,#4}}
\newcommand{\pc} {\beta_c}
\newcommand{\CAP} {C_{\mathrm{AP}}(\RR)} 
\newcommand{\x} {\ovl{x}}
\newcommand{\leqZ} {\leq_{\mathbb{Z}}}
\newcommand{\qR} {\ovl{\mathbb{R}}}
\newcommand{\lin} {\mathrm{l}}
\newcommand{\mc} {\mathrm{c}}
\newcommand{\Cstar} {C^*}
\newcommand{\prfl}[2] {\im[#1]\sqcup \TO^{|#2|}}

\begin{document} 

\title{Projective Structures in Loop Quantum Cosmology}
\author{Maximilian Hanusch\thanks{e-mail:
    {\tt hanuschm@fau.edu}}\\   
  \\
  {\normalsize\em Institut f\"ur Mathematik}\\[-0.15ex]
  {\normalsize\em Universit\"at Paderborn}\thanks{Now at Florida Atlantic University.}  
  \\[-0.15ex]
  {\normalsize\em Warburger Stra\ss e 100}\\[-0.15ex]
  {\normalsize\em 33098 Paderborn}\\[-0.15ex]
  {\normalsize\em Germany}}    
\date{May 5, 2015}
\maketitle

\frenchspacing
\begin{abstract}
  Projective structures have successfully been used for the construction of measures in the framework of loop quantum gravity. In the present paper, we establish such structures for the configuration space $\RR\sqcup \RB$, recently introduced in the context of homogeneous isotropic loop quantum cosmology. In contrast to the traditional space $\RB$, the first one is canonically embedded into the quantum configuration space of the full theory. In particular, for the embedding of states into a corresponding symmetric sector of loop quantum gravity, this is advantageous. However, in contrast to the traditional space, there is no Haar measure on $\RR\sqcup \RB$ defining a canonical kinematical $L^2$-Hilbert space on which operators can be represented. The introduced projective structures allow to construct a family of natural measures on $\RR\sqcup \RB$ whose corresponding $L^2$-Hilbert spaces we finally investigate.
\end{abstract}

\thispagestyle{empty}

\section{Introduction} 
In the framework of loop quantum gravity (LQG), measures are usually constructed by means of projective structures on the quantum configuration space of interest. For instance, the Ashtekar-Lewandowski measure arises in this way \cite{{ProjTechAL}}, and the same is true for the Haar measure on the Bohr compactification $\RB$ of $\RR$. \cite{Vel} This has been used as quantum configuration space for homogeneous isotropic loop quantum cosmology (LQC); a symmetry reduced version of LQG designed to describe the early universe near the Big Bang. \cite{MathStrucLQG} Unfortunately, there is no continuous embedding of $\RB$ into the quantum configuration space of LQG which additionally extends the embedding of the respective reduced classical configuration space. \cite{Brunnhack} This property, however, is crucial for the embedding approach for states formulated in \cite{BojoKa}.

Now, non-embeddability arises from the fact that, in contrast to the full theory, the cosmological quantum configuration space has been defined by means of linear curves instead of all the embedded analytic ones. \cite{Brunnhack} Thus, to overcome this problem, in \cite{ChrisSymmLQG} the embedded analytic curves were used to define the reduced quantum space as well; now being given by $\qR=\RR\sqcup \RB$. In particular, the embedding approach from \cite{BojoKa} here can be applied once a reasonable measure has been fixed. Now, since no Haar measure exists on $\qR$, \cite{InvConnLQG} such measures have to be constructed by hand.

In the present paper, we attack this issue by means of projective structures on $\qR$ which we then use to motivate the family of normalized Radon measures 
\begin{align}
\label{meashhh}
  \mu_{\rho,t}(A)=t\cdot \rho (\muL)(A\cap \RR)+ (1-t)\cdot \muB(A\cap \RB) \qquad \forall\: A\in \Borel\Rq
\end{align}
for $0\leq t\leq 1$ and $\rho(\muL)$ the push forward of the Lebesgue measure $\muL$ on $(0,1)$ by some homeomorphism $\rho\colon (0,1)\rightarrow \RR$. For this, we first reformulate the definition of a projective limit in a way more practicable for defining measures on compact Hausdorff spaces, such as, e.g., $\qR$. Then, we motivate a certain family of projective structures which will provide us with the measures \eqref{meashhh}. 
 As we will see, these measures give rise to only two different Hilbert space structures on $\qR$.
More precisely, up to \emph{canonical} isomorphisms, we will have the following three cases:
\begin{align*}
  \Lzw{\RR}{\muL},\qquad\quad \Lzw{\RR}{\muL}\oplus \Lzw{\RB}{\muB},\qquad\quad \Lzw{\RB}{\muB},
\end{align*}
whereby $\Lzw{\RR}{\muL}\oplus \Lzw{\RB}{\muB}$ and $\Lzw{\RB}{\muB}$ are isometrically isomorphic, just by dimensional arguments.
Anyhow, since $\Lzw{\RR}{\muL}$ is separable and $\Lzw{\RB}{\muB}$ is not so, there cannot exist any isometric isomorphism between these two spaces. 
\newline
\vspace{-2ex} 
\newline
This paper is organized as follows:
\begingroup
\setlength{\leftmargini}{14pt}
\begin{itemize}
\item
  \vspace{-3.8pt} 
  In Section 2, we fix the notations and provide a characterisation of projective limits convenient for defining measures. In Section 3, we briefly review some facts on invariant homomorphisms \cite{InvConnLQG} that we will need in the main part of this paper.
\item
  \vspace{-1ex}
  In Section 4, we first discuss some elementary properties of the space $\qR$. In particular, we prove a uniqueness result concerning the assumptions made in \cite{Jon} to the inner product on $C_0(\RR)\oplus \CAP$. 
  
  Then, we investigate how to write $\qR$ as projective limit, in order to construct reasonable Radon measures thereon. 
  Here, we discuss several possibilities, finally leading to the projective structures presented in the third part of Section 4. Basically, there we will use the fact that for each nowhere vanishing\footnote{Here $C_0(\RR)$ denotes the set of continuous functions on $\RR$ that vanish at infinity.} $f\in C_0(\RR)$ the functions 
  $\{f\}\sqcup \{\chi_l\}_{l\in \RR}$ with $\chi_l\colon x \mapsto \e^{\I l x}$ generate the $\Cstar$-algebra $C_0(\RR)\oplus \CAP$. Then, each such $f$ which is in addition injective will give rise to a projective structure similar to that one introduced in \cite{Vel} for the space $\RB$. In the last part, we finally use these structures to construct a family of normalized Radon measures on $\qR$ which we then show to define two different non-isomorphic $L^2$-Hilbert spaces on $\qR$. 
\end{itemize}
\endgroup
\section{Preliminaries}
\label{sec:prelimin}
We start this section by fixing the notations. Then, we give a short introduction into projective structures on compact Hausdorff spaces and consistent families of normalized Radon measures.
\subsection{Notations}
\label{subse:Notations}
A curve $\gamma$ in a manifold $M$ is a continuous map $\gamma\colon I\rightarrow M$ for $I\subseteq \RR$ an interval, i.e., of the form $[a,b],[a,b),(a,b]$ or $(a,b)$ for $a<b$. Then, the curve $\gamma$ is said to be of class $C^k$ iff $M$ is a $C^k$-manifold and iff there is a $C^k$-curve $\gamma'\colon (a',b')\rightarrow M$ with $I\subseteq (a',b')$ and $\gamma'|_{I}=\gamma$. By a path, we will understand a curve which is $C^\infty$ or analytic ($C^\w$) and defined on some closed interval. 

Let $(P,\pi,M,S)$ be a principal fibre bundle with total space $P$, base manifold $M$, projection map $\pi\colon P\rightarrow M$ and structure group $S$. 
If $\w$ is a smooth connection on $P$, $\gamma\colon [a,b]\rightarrow M$ a path and $p\in \pi^{-1}(\gamma(a))$, then $\gamma_p^\w\colon [a,b]\rightarrow P$ denotes the \emph{horizontal lift} 
of $\gamma$ w.r.t.\ $\omega$ in the point $p$. 
The isomorphism 
\begin{align*}
	\parall{\gamma}{\omega}\colon \pi^{-1}(\gamma(a))\rightarrow \pi^{-1}(\gamma(b)),\quad 
p\mapsto \gamma_p^\w(b)
\end{align*}
is called \emph{parallel transport} along $\gamma$ w.r.t.\ $\omega$. In particular, 
here isomorphism means that $\parall{\gamma}{\omega}(p\cdot s)=\parall{\gamma}{\omega}(p)\cdot s$ holds for all $p\in F_{\pi(p)}$ and all $s\in S$. Finally, $\Con$ will denote the set of smooth connections on $P$.

For an abelian unital $\Cstar$-algebra $\aA$, we denote by $\Spec(\aA)$ the set of non-zero multiplicative, $\mathbb{C}$-valued functionals on $\aA$, equipped with usual Gelfand-topology.  
The Gelfand transform $\widehat{a}\in C(\Spec(\aA))$ of $a\in \aA$ is given by $\widehat{a}\colon \Spec(\aA)\rightarrow \CCC$, $\psi\mapsto\psi(a)$.

Finally, we declare the following
\begin{convention}
  \label{conv:Boundedfunc}
  Let $X$ be some set. Then, 
  \begingroup
  \setlength{\leftmargini}{14pt}
  \begin{itemize}
  \item
    \vspace{-3pt}
    By $B(X)$ we will denote the set of bounded functions on $X$.
    \item
    \itspacecc
    For $\aA\subseteq B(X)$ a $\Cstar$-subalgebra and $\ishomc\colon Y\rightarrow X$ some map, we let $\ishomc^*\aA$ denote the $\Cstar$-subalgebra of $B(Y)$ generated by the functions $\ishomc^*f$ for $f\in \aA$. 
  \item  
    \itspacecc
    For $\aA\subseteq B(X)$ a $\Cstar$-algebra, we let $X_\aA$ denote the set of all $x\in X$ for which
    \begin{align*}  
      \iota_X\colon X &\rightarrow \mathrm{Hom}(\aA,\mathbb{C})\\ x&\mapsto [f\mapsto f(x)]
    \end{align*}
    is non-zero, i.e., $X_\aA=\{x\in X\: |\:\exists\: f\in \aA : f(x)\neq 0\}$. This means that $x\in X_\aA$ iff $\iota_X(x)\in \Spec(\aA)$, and it is not hard to see that $\iota_X(X_\aA)$ is dense in $\Spec(\aA)$. \cite{Rendall}, \cite{ChrisSymmLQG}, \cite{InvConnLQG}.
  \item
    \itspacecc
    Motivated by this denseness property, the spectrum of $\aA$ is denoted by $\ovl{X}$ in the following.
  \item	
    \itspacecc
    If $X$ is a locally compact Hausdorff space, then $C_0(X)$ denotes the set of continuous functions on $X$ that vanish at infinity.
  \item
    \itspacecc
    By $\CAP$, we will denote the almost periodic functions on $\RR$. This is the $\Cstar$-subalgebra of $B(\RR)$ generated by the set $\RRD$ of continuous characters $\chi_l\colon \RR \rightarrow \TO$, $x\mapsto \e^{\I l x}$. 
    Here, and in the following, $\TO$ denotes the unit circle ($1$-Torus) $\{z \in \mathbb{C}\:|\: |z|=1\}$.
  \item
    \itspacecc
    We define the Bohr compactification $\RB$ of $\RR$ to be given by $\Spec(\CAP)$. 
    
    Let $\DDD$ denotes the set of all $^*$-homomorphisms $\psi \colon \RRD\rightarrow \TO$. Then, it follows from Subsection 1.8 in \cite{RudinFourier} that the restriction map $\res\colon \RB\rightarrow \DDD$, $\psi\mapsto \psi|_{\RRD}$ is bijective.\footnote{See, e.g., Lemma 3.8 in \cite{Thesis}.} In particular, this means that, in order to define an element of $\RB$, it suffices to determine its values on $\RRD$. 
    
    Then, $\muB$ will denote the Haar measure on $\RB$ that corresponds to the continuous abelian group structure 
    \begin{align*}
      \psi + \psi' := \res^{-1}\left(\res(\psi)\cdot\res(\psi')\right)\qquad -\psi:=\res^{-1}\big(\hs\ovl{\res(\psi)}\hs\big)\qquad\:\: e:=\res^{-1}(1)
    \end{align*}
    for $\psi,\psi' \in \RB$. Here, $(\zeta\cdot\zeta')(x):=\zeta(x)\cdot\zeta'(x)$, $\ovl{\zeta}(x):=\ovl{\zeta(x)}$ and $1(x):=1$ 
     for all $x\in \RR$ and $\zeta,\zeta'\in \DDD$.  
    \hspace*{\fill}{\footnotesize$\lozenge$}
  \end{itemize}
  \endgroup
\end{convention}
\noindent
Finally, if $G$ is a group, $X$ a set and $\varphi \colon G\times X\rightarrow X$ a left action, we define 
\begin{align*}
	\Stab{\varphi}{x}:=\{g\in G\:|\:\varphi(g,x)=x\}
\end{align*}
 as well as $\varphi_g\colon X\rightarrow X$,\:\: $x\mapsto \varphi(g,x)$. If $G$ is a Lie group and $g\in G$, then 
\begin{align*}
	\alpha_g\colon G\rightarrow G,\quad h\mapsto g\cdot h\cdot g^{-1}
\end{align*} 
 denotes the conjugation by $g$ and $\Add{g}\colon \mathfrak{g}\rightarrow \mathfrak{g}$ its differential $\dd_e\alpha_g$ at $e\in G$.
\subsection{Projective Structures and Radon measures}
\label{subsec:ProjStruc} 
\begin{definition}
  \label{def:ProjLim}
  Let $\{X_\alpha\}_{\alpha\in I}$ be a family of compact Hausdorff spaces for $(I,\leq)$ a directed set.\footnote{This means that $\leq$ is a reflexive and transitive relation on $I$, and that for each two $\alpha,\alpha'\in I$ we find some $\alpha''\in I$ such that $\alpha,\alpha'\leq \alpha''$ holds.} A compact Hausdorff space $X$ is called projective limit of $\{X_\alpha\}_{\alpha\in I}$ iff 
  \begin{enumerate}
  \item
    For each $\alpha\in I$, there is a continuous surjective map $\pi_\alpha\colon X\rightarrow X_\alpha$.
  \item
    For $\alpha_1,\alpha_2 \in I$ with $\alpha_1\leq \alpha_2$, there is a continuous map $\pi^{\alpha_2}_{\alpha_1}\colon X_{\alpha_2}\rightarrow X_{\alpha_1}$ for which $\pi^{\alpha_2}_{\alpha_1}\cp \pi_{\alpha_2}=\pi_{\alpha_1}$ holds. 
    
    It follows that each of these maps is surjective and that $\pi^{\alpha_2}_{\alpha_1}\cp\pi^{\alpha_3}_{\alpha_2}=\pi^{\alpha_3}_{\alpha_1}$ holds if  $\alpha_1\leq \alpha_2\leq\alpha_3$ for $\alpha_1,\alpha_2,\alpha_3\in I$.
  \item
    If $x,y\in X$ with $x\neq y$, there is some $\alpha\in I$ with $\pi_\alpha(x)\neq \pi_\alpha(y)$.
  \end{enumerate}
\end{definition}
It is proven in Lemma \ref{lemma:equivalence} that the above definition of a projective limit is equivalent to the usual definition \cite{ProjTechAL} as a subset
 \begin{align*}  
    \widehat{X}=\left\{\hx \in \textstyle\prod_{\alpha\in I}X_\alpha\: \:\big|\:\: \pi_{\alpha_1}^{\alpha_2}(x_{\alpha_2})=x_{\alpha_1}\:\: \forall\:\alpha_1\leq \alpha_2\right\}\subseteq\textstyle\prod_{\alpha \in I}X_\alpha
  \end{align*} 
 of the Tychonoff product $\prod_{\alpha \in I}X_\alpha$. In particular, each two projective limits of a fixed family of compact Hausdorff spaces are homeomorphic if the same transition maps are used. 
 
Anyhow, as it provides us with more flexibility, in the following we will use Definition \ref{def:ProjLim} instead of the Cartesian product version. 
\begin{definition}[Borel, Radon Measures]
  \label{def:Mass}
  Let $Y$ be a Hausdorff space and $\mathfrak{B}(Y)$ the Borel $\sigma$-algebra of $Y$.
  \begingroup
  \setlength{\leftmargini}{14pt}
   \begin{itemize}
  \item
    \label{def:Mass1}
    A Borel measure $\mu$ on $Y$ is a locally finite\footnote{This means that for each $y\in Y$ we find $U\subseteq Y$ open with $y\in U$ and $\mu(U)<\infty$.} measure $\mu\colon \mathfrak{B}(Y)\rightarrow [0,\infty]$. It is said to be normalized if $\|\mu\|:=\mu(Y)=1$ holds.  
  \item
    \label{def:Mass2}
    A Borel measure $\mu$ on $Y$ is called inner regular iff for each $A\in \mathfrak{B}(Y)$ we have 
	\begin{align*}    
    	\mu(A)=\sup\{\mu(K):K \text{ is  compact and }K\subseteq A\}.
    \end{align*} 
  \item
    \label{def:Mass3}
    A Radon measure $\mu$ is an inner regular Borel measure. It is called finite iff     
    $\mu(Y)<\infty$ holds. Recall that each finite Radon measure is outer regular (hence regular), i.e., for each $A\in \Borel(Y)$ we have
    \begin{align*} 
    	\mu(A)=\inf\{\mu(U):U \text{ is  open and }A\subseteq U\}.
    \end{align*}
  \item
    Assume that we are in the situation of Definition \ref{def:ProjLim}, and that $\{\mu_\alpha\}_{\alpha\in I}$ is a family of Radon measures $\mu_\alpha\colon \mathfrak{B}(X_\alpha)\rightarrow [0,\infty]$. Then, $\{\mu_\alpha\}_{\alpha\in I}$ is called consistent iff $\mu_{\alpha_1}$ equals the push forward measure $\pi^{\alpha_2}_{\alpha_1}(\mu_2)$ whenever $\alpha_1\leq \alpha_2$ holds for $\alpha_1,\alpha_2\in I$. 
  \end{itemize}
  \endgroup
\end{definition}
\begin{lemma}
  \label{lemma:normRM}
  Let $X$ and $\{X_\alpha\}_{\alpha\in I}$ be as in Definition \ref{def:ProjLim}. Then, the normalized Radon measures on $X$ are in bijection with the consistent families of normalized Radon measures on $\{X_\alpha\}_{\alpha\in I}$.
\end{lemma}
\begin{proof}
  See Lemma \ref{lemma:ConstMeas}.
\end{proof} 

\section{Quantum Configuration Spaces in LQG}
\label{sec:GenConInv}
In this section, we will give a short introduction into invariant generalized connections and homomorphisms of paths. 
For simplicity, here we restrict to the case of trivial principal fibre bundles; the general results can be found in \cite{InvConnLQG}. In the last part of this section, we will discuss the 
case of homogeneous isotropic loop quantum gravity. 
\subsection{Generalized Connections and Invariance}
\label{subsec:GenConInvvvv}
Let $P=M\times S$ be a trivial principal fibre bundle with base manifold $M$ and compact structure group $S$. Moreover, let $\Paths$ be a fixed 
set of paths in $M$.

 For $\gamma\in \Paths$ with $\dom[\gamma]=[a,b]$, we define 
\begin{align}
\label{eq:patralla} 
 h_\gamma\colon \Con \rightarrow S,\quad \w \mapsto \left(\pr_2\cp \parall{\gamma}{\w}\right)(\gamma(a),e)
\end{align} 
 and denote by $\gG$ the $^*$-algebra generated by all functions of the form $f\cp h_\gamma$ with $f\in C(S)$ and $\gamma\in \Paths$. Due to compactness of $S$, we have $\gG\subseteq B(\Con)$, so that we can define the $\Cstar$-algebra of cylindrical functions $\cC$ to be closure of $\gG$ in $B(\Con)$. 
The spectrum of $\cC$ is denoted by $\ovl{\Con}$, and its elements are called generalized connections in the following. 

Let $(G,\LGA)$ be a Lie group of automorphisms of $P$, i.e., a Lie group $G$ together with a smooth left action $\LGA\colon G\times P\rightarrow P$ with
$\LGA(g,p\cdot s)=\LGA(g,p)\cdot s$ for all $p\in P$, $g\in G$ and $s\in S$.
Then, $\LGA$ gives rise to the following two further left actions:
\begingroup
\setlength{\leftmargini}{16pt}
\begin{itemize}
\item
  \vspace{-0.3pt}
  $\INDA\colon G\times M\rightarrow M$, $(g,m)\mapsto \pi(\LGA(g,p_m))$ for $p_m\in F_m$ arbitrarily chosen,
\item
  \vspace{-1pt}
  $\conac\colon G\times \Con \rightarrow \Con$, $(g,\w)\mapsto \LGA_{g^{-1}}^*\w$.
\end{itemize}
\vspace{-1pt}
\endgroup
\noindent
The set of invariant connections is defined by 
\begin{align*}
\ConG:=\{\w \in \Con\:|\: \Stab{\conac}{\w}=G\}
\end{align*}
whereby, for later convenience, we will assume $\ConG$ to be given by the image of some injective map $\ishomc\colon \IR\hookrightarrow \Con$.

Now, if $\Paths$ is invariant in the sense that $\gamma'(t):=\INDA(g,\gamma(t))$ is in $\Paths$ for all $\gamma\in \Paths$ and all $g\in G$, by Corollary 3.8 in \cite{InvConnLQG} $\conac$ can be uniquely extended to an action on $\A$. More precisely, 
there exists a unique left action $\Asp\colon G \times \A\rightarrow \A$ (namely $\Asp\colon (g,\wq)\mapsto \wq\cp \phi_g^*$), such that
\begingroup
\setlength{\leftmargini}{16pt}
\begin{itemize}
\item
  \vspace{-0.3pt}
	$\Asp_g$ is continuous for all $g\in G$,
\item
  \vspace{-1pt}
	$\iota_\Con \cp \conac_g = \Asp_g \cp \iota_\Con$, i.e., the following diagram commutes for each $g\in G$:
	\vspace{-8pt}
	
\begin{center}
  \makebox[0pt]{
    \begin{xy}
      \xymatrix{
        \ovl{\Con} \ar@{->}[r]^-{\Phi_g}   &  \ovl{\Con}   \\
        \Con\ar@{.>}[u]^{\iota_\Con}\ar@{->}[r]^-{\phi_g} & \Con  \ar@{.>}[u]^{\iota_\Con}
      }
    \end{xy}
  }
\end{center} 
 Here, $\iota_\Con$ denotes the map from Convention \ref{conv:Boundedfunc}.   
\end{itemize}
  \vspace{-4pt}
\endgroup
\noindent
In analogy to $\ConG$, one can now define the set of invariant generalized connections by
\begin{align*}
	\AQR:=\{\wq \in \A\:|\:\Stab{\Phi}{\wq}=G\}.
\end{align*}
The space $\AQR$ is compact and \cite{InvConnLQG}
\begin{align*}
  \Spec(\ishomc^*\cC)\cong \ARQ:=\ovl{\iota_\Con(\ConG)}\subseteq \AQR
\end{align*}
holds. Here, the first homeomorphism is just given by
\begin{align}
\label{embext}
	\oishomc\colon  \Spec(\ishomc^*\cC)\rightarrow  \ARQ\subseteq \A,\quad\psi \mapsto [f\mapsto \psi(f\cp \ishomc)].
\end{align}

\begin{center}
  \makebox[0pt]{
    \begin{xy}
      \xymatrix{
        & \Spec(\ishomc^*\cC)\ar@{->}[r]^-{\oishomc}_-{\cong}& \ovl{\ConG} & \hspace{-26pt}\subseteq\AQR \subseteq \A  
        \\
        &  \IR \ar@{->}[r]^-{\ishomc}  \ar@{->}[u]^-{\iota_{\IR}}  & \ConG \ar@{->}[u]^-{\iota_{\Con}}& \hspace{-60pt}\subseteq\Con.  
      }
    \end{xy}
  }
\end{center}   
Then,\footnote{If a normalized Radon measure on $\Spec(\ishomc^*\cC)$ has been fixed, the corresponding $L^2$-Hilbert space $\Hil$ is canonically embedded via $\mj\colon \psi\mapsto \langle\psi,\cdot\rangle$ into the algebraic (topological) dual of $C(\Spec(\ishomc^*\cC))$. In fact, by regularity of finite Radon measures, $C(\Spec(\ishomc^*\cC))$ is dense in $\Hil$ so that injectivity of $\mj$ is clear because then $\mj(\psi)=0$ already implies that $\langle\psi, \psi\rangle=0$. Moreover, finiteness of $\|\mj(\psi)\|_{\mathrm{op}}$, i.e., continuity of $j(\psi)$ is immediate from the Cauchy-Schwarz inequality.} distributional states, 
i.e., elements $L$ of the algebraic (topological) dual of $C(\Spec(\ishomc^*\cC))$ can be embedded into the algebraic (topological) dual of $C(\A)$ by \cite{BojoKa}
\begin{align}
\label{eq:statesemb}
	L\mapsto \big[ \widehat{f} \mapsto L\big(\hs\widehat{f}\cp \oishomc\:\big)\big]\qquad \forall\: \widehat{f}\in C(\A).
\end{align}
In fact, $\widehat{f}$ is the Gelfand transform of some $f\in \cC$, hence $\widehat{f}\cp \oishomc=\widehat{f\cp \ishomc}\in C(\Spec(\ishomc^*\cC))$. Thus, \eqref{eq:statesemb} is well defined and injective by definition of $\ishomc^*\cC$. 
\subsection{Homomorphisms of Paths}
Let $P=\RR^3\times \SU$ and $\Paths$ denote the set of the linear or embedded analytic\footnote{A curve $\gamma\colon [a,b]\rightarrow M$ is said to be embedded analytic iff there is an analytic embedding $\gamma'\colon  (a',b')\rightarrow \mathbb{R}^3$ such that $[a,b]\subseteq (a',b')$ and $\gamma=\gamma'|_{[a,b]}$. Here, an embedding means an immersion which is a homeomorphism onto its image (equipped with the relative topology).} curves in $\RR^3$. Recall that two paths $\gamma_1,\gamma_2\in \Paths$ are said to be (holonomically) equivalent (write $\gamma_1\sim_\Con \gamma_2$) iff $\Paths_{\gamma_1}^\w=\Paths_{\gamma_2}^\w$ holds for all $\w\in \Con$.

Now, let $\Paths\ni \gamma\colon [a,b]\rightarrow M$. Then,
\begingroup
\setlength{\leftmargini}{14pt}
\begin{itemize}
\item
  \vspace{-0.8ex}
  The inverse curve of $\gamma$ is defined by $\gamma^{-1}\colon [a,b]\ni t\mapsto \gamma(b+a-t)$.  
\item
  \vspace{-0.8ex}
  A decomposition of $\gamma$ is a family of curves $\{\gamma_i\}_{0\leq i\leq k-1}$ 
  such that $\gamma|_{[\tau_i,\tau_{i+1}]}=\gamma_i$ for $0\leq i\leq k-1$ and real numbers $a=\tau_0<{\dots}<\tau_k=b$. 
\end{itemize}
\endgroup
\noindent
Obviously, $\Paths$ (chosen as above) is stable under inversion and decomposition of its elements; and we define the set $\Homm(\Paths, \SU)$ of Homomorphisms of paths as follows:\footnote{This definition differs from the usual one \cite{Ashtekar2008} in the point that we require $\hommm$ to be compatible w.r.t.\ decompositions of paths and not w.r.t.\ their concatenations. This helps to avoid unnecessary technicalities as it allows to restrict to embedded analytic curves instead of considering all the piecewise ones.}

An element $\hommm\in \Homm(\Paths, \SU)$ 
is a map $\hommm \colon \Paths\rightarrow \SU$ with
\begingroup
\setlength{\leftmargini}{14pt}
\begin{itemize}
\item
  \vspace{-0.8ex}
  $\hommm(\gamma^{-1})=\hommm(\gamma)^{-1}$\: and\: $\hommm(\gamma)= \hommm(\gamma_{k-1}) \cdot {\dots} \cdot\hommm(\gamma_{0})$\: for each decomposition\: $\{\gamma_i\}_{0\leq i\leq k-1}$ of $\gamma\in \Paths$,
\item
  \vspace{-0.8ex}
  $\hommm(\gamma)=\hommm(\gamma')$ if $\gamma\sim_\Con \gamma'$ for  $\gamma,\gamma'\in \Paths$.
\end{itemize}
\endgroup
\vspace{-0.3ex}
\noindent
In particular, for each $\w\in\Con$ the map $\gamma \mapsto h_\gamma(w)$, cf.\ \eqref{eq:patralla}, is such a homomorphism. 

Now, due to denseness of 
$\iota_\Con(\Con)$ in $\ovl{\Con}$, for each $\ovl{\w}\in \ovl{\Con}$ we find some net $\{\w_\alpha\}_{\alpha\in I}\subseteq \Con$ with $\{\iota_\Con(\w_\alpha)\}_{\alpha\in I}\rightarrow \w$, and it follows (cf.\ Appendix B in \cite{InvConnLQG}) that the map 
\begin{align}
  \label{eq:SpeczuHomm}
  \begin{split}
    \homiso\colon \ovl{\Con} &\rightarrow \Homm(\Paths,\SU)\\
    \ovl{\w}&\mapsto \left[\gamma \mapsto \lim_\alpha h_\gamma(\w_\alpha)\right] 
  \end{split}
\end{align}
is a well-defined bijection. Moreover, we have   
\begin{align}
  \label{eq:pigammawquer}
  \homiso(\ovl{\w})(\gamma)=\left(\hs\ovl{\w}\big([h_\gamma]_{ij}\big)\right)_{ij}\in \SU\qquad \forall\:\ovl{\w}\in \ovl{\Con}, 
\end{align}
where for $1\leq i,j\leq 2$ and $s\in \SU$ by $[s]_{ij}$ we mean the respective matrix entry. 
Finally, if $(G,\Theta)$ is a Lie group of automorphisms of $P$, we define the corresponding set of invariant homomorphisms by $\Homm_\inv(\Paths,\SU):=\homiso(\AQR)$.

\subsection{Homogeneous Isotropic Loop Quantum Cosmology}
\label{subsec:HogELQC}
Let $P=\RR^3\times \SU$ and $\Ge:=\Gee$ with $\varrho\colon \SU \rightarrow \SO$ the universal covering map given by
\begin{align*}
\varrho(\sigma)= \murs^{-1} \cp \Add{\sigma}\cp\: \murs.
\end{align*}
 Here, $\murs\colon \RR^3\rightarrow\su$ denotes the linear map with $\murs(\vec{e}_i)=\tau_i$ for $i=1,2,3$  
whereby $\{\vec{e}_1,\vec{e}_2,\vec{e}_3\}$ means 
the standard basis in $\mathbb{R}^3$ 
and
\begin{align*}
  \tau_1:=\begin{pmatrix} 0 & -\I  \\ -\I & 0  \end{pmatrix}\qquad\qquad \tau_2:=\begin{pmatrix} 0 & -1  \\ 1 & 0  \end{pmatrix}\qquad\qquad \tau_3:=\begin{pmatrix} -\I & 0  \\ 0 & \I  \end{pmatrix}.
\end{align*} 
Now, let $\Theta\colon G\times P\rightarrow P$ be given by $\Theta(g,p):=g \cdot_{\varrho} p$.  
Then, the corresponding set of invariant connections is parametrized by $\RR$, 
via $\ishomc\colon \RR\rightarrow \ConG\subseteq \Con$
\begin{align*}
  \ishomc(c)_{(x,s)}(\vec{v}_x,\vec{s}_s):= c \Add{s^{-1}}[\hs\murs(\vec{v}_x)\hs]+s^{-1}\vec{s}_s\qquad \text{for}\qquad (\vec{v}_x,\vec{s}_s)\in T_{(x,s)}P. 
\end{align*}
Let $\cC_{\lin}$ and $\cC_\w$ denote the $\Cstar$-algebras of cylindrical functions that correspond to the sets of linear and embedded analytic curves, respectively. Then, 
\begin{align*}
  \ishomc^*\cC_\lin\stackrel{\text{\hspace{-1pt}\cite{MathStrucLQG}}}{=} \CAP  \qquad\quad\text{and}\qquad\quad \ishomc^*\cC_\w \stackrel{\text{\hspace{-1pt}\cite{ChrisSymmLQG}}}{=} C_0(\RR)\oplus \CAP,\footnotemark
\end{align*}
\footnotetext{Let $\aA$ denote the set of all continuous, bounded function on $\RR$ which can be written as a sum $f_0 + f_{\mathrm{AP}}$ for $f_0\in C_0(\RR)$ and $f_{\mathrm{AP}}\in \CAP$. Then, Corollary B.2 in Version 2 of \cite{ChrisSymmLQG} shows that $\aA$ is a $\Cstar$-algebra and that $\aA=C_0(\RR)\oplus \CAP$ holds. Here, $\oplus$ means the direct sum of vector spaces not of $\Cstar$-algebras.}
giving rise to the spectra\footnote{The topology on $\qR$ will be specified in Subsection \ref{subsec:Rquer}.}
\vspace{-4pt}
\begin{align*}
	\hspace{-2pt}\Spec(\ishomc^*\cC_\lin)=\RB\qquad\quad\text{and}\qquad\quad\Spec(\ishomc^*\cC_\w)\stackrel{\text{\hspace{-1pt}\cite{ChrisSymmLQG}}}{\cong} \RR\sqcup \RB=:\qR,
\end{align*}
respectively. Originally, $\Spec(\ishomc^*\cC_\lin)=\RB$ has been used as cosmological quantum configuration space. This space, however, here cannot be compatibly embedded into the quantum configuration space $\ovl{\Con}:=\Spec(\cC_\w)$ of the full theory \cite{Brunnhack}, i.e., there is no embedding $\ovl{\ishomc}\colon \RB\rightarrow \A$  
which makes the following diagram commutative:
\begin{center}
  \makebox[0pt]{
    \begin{xy}
      \xymatrix{
        & \RB\ar@{->}[r]^{\ovl{\ishomc}}& \A
        \\
        &  \RR \ar@{->}[r]^-{\ishomc}  \ar@{->}[u]^-{\iota_R}  & \Con \ar@{->}[u]^-{\iota_{\Con}}.  
      }
    \end{xy}
  }
\end{center}   
This is exactly because $\ishomc^*\cC_\lin \neq \ishomc^*\cC_\w$ holds, cf.\ Theorem 2.20 in \cite{ChrisSymmLQG}, so that \eqref{embext} is not well defined in this case. In particular, the embedding approach \cite{BojoKa} for states, we have discussed in the end of Subsection \ref{subsec:GenConInvvvv}, cannot be applied. 
For this reason, in \cite{ChrisSymmLQG} 
 the space $\Spec(\ishomc^*\cC_\w)$ was introduced. Indeed, here \eqref{embext} exists and is even uniquely determined by its extension property.

Now, it follows \cite{InvConnLQG} that for $\epsilon \in \Homm(\Paths,\SU)$ we have 
$\hommm\in \Homm_\inv(\Paths,\SU)$ iff
\begin{align} 	
  \label{eq:InvGenConnRel}
  \hommm(v+\uberll{\sigma}{\gamma)}=(\Co{\sigma}\cp \hommm)(\gamma)\quad\qquad \forall\:(v,\sigma)\in \Gee,\quad \forall\:\gamma \in \Paths.
\end{align}
This has the following consequences for the quantum-reduced spaces $\AQR$ that correspond to the two choices of sets of curves $\Paths$ we have discussed so far: \cite{InvConnLQG}

\begin{center}
\begin{tabular}{c|cl}
	curves  &  $\ARQ$ & $\AQR$ \\[0.5pt]\hline\rule{0pt}{12pt}
	\text{linear}          &  $\phantom{\RR\sqcup}\RB$ & homeomorphic to $\RB$ \\
	\text{embedded analytic}  &  $\RR\sqcup\RB$ &  larger than $\RR\sqcup\RB$
\end{tabular}
\end{center}
Moreover, if
\begingroup
\setlength{\leftmargini}{14pt}
\begin{itemize}
\item
$\TO_{\vec{n}}:=\{\exp(t \cdot \murs(\vec{n}))\:|\:t\in \RR\}$ 
denote the maximal torus in $\SU$ which corresponds to $\vec{n}\in \RR^3$,
\item 
$\gamma\in \Paths$ is equivalent to a linear curve of the form $[0,l]\ni t\mapsto x + t\cdot \vec{n}$,
\end{itemize}
\endgroup
\noindent
then for each $\hommm\in \Homm_\inv(\Paths,\SU)$ we have 
  $\hommm(\gamma)\in \TO_{\vec{n}}$. In particular, $\ARQ$ and $\AQR$ are both of measure zero w.r.t.\ the Ashtekar-Lewandowski measure $\mAL$ (see, e.g., Appendix \ref{sec:AshLewMeasure}) on $\A$ because for 	  
  $\pi_\gamma\colon \AQR\rightarrow \SU$,\:\: $\w \mapsto \homiso(\w)(\gamma)$, we have 
  \begin{align}
    \label{eq:meas0}
    \mAL(\AQR)\leq \mAL\left(\pi_\gamma^{-1}\left(\TO_{\vec{n}}\right)\right)=\mu_1\left(\TO_{\vec{n}}\right)=0
  \end{align}
  for $\mu_1$ the Haar measure on $\SU$.

\section{$\boldsymbol{\qR}$ as a Projective Limit}   
\label{sec:projlim}
In the first part of this section, we will discuss the elementary properties of the cosmological quantum configuration space $\qR$. \cite{ChrisSymmLQG} Then, we will motivate certain projective structures on $\qR$ which we use to fix a family of natural Radon measures on this space. Finally, we will investigate the $L^2$-Hilbert spaces these measures define. 

So, in what follows, let $\Paths$ denote the set of embedded analytic curves in $\RR^3$. Moreover, let $P=\RR^3\times \SU$ and $\Ge=\Gee$ the Lie group with  action $\Theta \colon \Ge \times P\rightarrow P$,\:\:$(g,p)\mapsto g\cdot_\varrho p$ defined in Subsection \ref{subsec:HogELQC}. 

\subsection{Topological Aspects}
\label{subsec:Rquer} 
As already mentioned in Subsection \ref{subsec:HogELQC}, using the set of embedded analytic curves to define the configuration space of homogeneous isotropic LQC leads to the spectrum of the $\Cstar$-algebra $\ishomc^*\cC_\w=C_0(\RR)\oplus \CAP$. \cite{ChrisSymmLQG} 
Now, if we equip $\qR:= \RR \sqcup \RB$ with the topology generated by the sets of the following types \cite{ChrisSymmLQG} 
\begin{align*}
  \begin{array}{lcrclcl}
    \textbf{Type 1:} && V & \!\!\!\sqcup\!\!\! & \emptyset 
    && \text{with open $V \subseteq \RR$} \\
    \textbf{Type 2:} && K^c & \!\!\!\sqcup\!\!\! & \RB
    && \text{with compact $K \subseteq \RR$} \\[-1.5pt]
    \textbf{Type 3:} && f^{-1}(U) & \!\!\!\sqcup\!\!\! & \widehat{f}^{-1}(U) 
    && \text{with open $U \subseteq \mathbb{C}$ and $f \in \CAP$},
  \end{array}
\end{align*}
Proposition 3.4 in \cite{ChrisSymmLQG} states that $\Spec(\ishomc^*\cC_\w)\cong \RR\sqcup \RB$ via the homeomorphism 
\begin{align*}
\xi\colon \RR\sqcup \RB\rightarrow \Spec(\ishomc^*\cC_\w)
\end{align*}
 defined by
\begin{equation}
  \label{eq:Ksiii}
  \xi(\x) := 
  \begin{cases} 
    \hspace{27pt}f\mapsto f(\x) &\mbox{if } \x\in \RR\\ 
    f_0\oplus f_{\mathrm{AP}}\mapsto \x(f_{\mathrm{AP}})  & \mbox{if } \x\in \RB
  \end{cases}
\end{equation}
for $f_0\in C_0(\RR)$ and $f_{\mathrm{AP}}\in \CAP$. 
It is straightforward to see that the subspace topologies of $\RR$ and $\RB$ w.r.t.\ the above topology coincide with their usual ones, and obviously we have 
\begin{align}
\label{eq:iotaxiR}
	\xi(x) = \iota_\RR(x)\qquad \forall x\in \RR.
\end{align} 
Now, to fix a reasonable measure on $\qR$ which provides a kinematical Hilbert space on which the dynamics of the reduced theory can be defined, we have the following possibilities:
\begingroup
\setlength{\leftmargini}{18pt}
\begin{enumerate}
\item[{\bf 1)}]
  Mimicking the space $\RB$, we can try to define a group structure on $\qR$ continuous w.r.t.\ the above topology. This would provide us with a natural Haar measure on $\qR$, but is impossible as no such group structure exists. \cite{InvConnLQG}
\item[{\bf 2)}]
  Since $\qR$ is canonically embedded into the quantum configuration space $\ovl{\Con}$, it is measurable w.r.t.\ the Ashtekar-Lewandowski measure $\mAL$ on $\ovl{\Con}$. However, due to \eqref{eq:meas0}, restricting $\mAL$ to $\qR$ only gives the zero-measure.
\item[{\bf 3)}]
  We can construct a projective structure on $\qR$ in order to fix reasonable finite Radon measure thereon.
  This seems to be the most canonical approach as both the Ashtekar-Lewandowski measure on $\ovl{\Con}$ and the Haar measure on $\RB$ \cite{Vel} arise in this way.
\end{enumerate}
\endgroup
\noindent
So, we will follow up the third approach by making use of the following straightforward result from \cite{InvConnLQG} 
(see also Conclusions in \cite{ChrisSymmLQG}) characterizing the finite Radon measures on $\qR$.
\begin{lemma}
  \label{lemma:Radon}
  \begin{enumerate}
  \item	
  \label{lemma:Radon1}
    We have $\BRq=\Borel(\hs\RR)\sqcup\Borel(\hs\RB)$.
  \item
  \label{lemma:Radon2}
    If $\mu$ is a finite Radon measure on $\BRq$, then $\mu|_{\Borel(\mathbb{R})}$ and $\mu|_{\Borel(\RB)}$ are finite Radon measures as well. Conversely, if $\mu_\RR$ and $\mu_{\B}$ are finite Radon measures on $\Borel(\mathbb{R})$ and $\Borel(\RB)$, respectively, then  $\mu_\RR\oplus \mu_{\B}$ defined by  
    \begin{align}
      \label{eq:RadonMeasures}	
      (\mu_\RR\oplus \mu_{\B})(A):=\mu_\RR(A\cap \Borel(\mathbb{R}))+ \mu_{\B}(A\cap \Borel(\RB))\qquad \text{for}\qquad A\in \BRq	
    \end{align}
    is a finite Radon measure on $\BRq$.
  \end{enumerate}
  \begin{proof}
    \begin{enumerate}
  \item	
     The right hand side is a $\sigma$-algebra and contains the $\sigma$-algebra on the left hand side because $U\cap \RR\in \Borel(\hs\RR)$ as well as $U\cap \RB \in \Borel(\hs\RB)$ for $U\subseteq \qR$ open. Now, $\Borel(\hs\RR)\subseteq \BRq$ because if $U\subseteq \RR$ is open, it is open in $\qR$ as well. Finally, $\Borel(\hs\RB)\subseteq \BRq$ because if $A\subseteq \RB$ is closed, it is compact, hence compact (closed) in $\qR$.
  \item
  	It is clear that, under the given assumptions, the restrictions and the sum are finite Borel measures. Since compactness in some subspace topology is the same as compactness in the inducing one, 
  	inner regularity of the restriction measures is clear; that of the sum follows by a simple $\epsilon/2$-argument.
  \end{enumerate}
  \end{proof}
\end{lemma}
By the above lemma, each normalized Radon measure on $\qR$ can be written as
  $t\:\mu_1\hs \oplus\hs (1-t)\:\mu_2$ 
for $t\in [0,1]$ and normalized\footnote{Of course, if $t=0$ or $t=1$, it doesn't matter which normalized Radon measure we choose on $\RR$ or $\RB$, respectively. Thus, in these cases we also allow $\mu_1=0$ and $\mu_2=0$, respectively.}  Radon measures $\mu_1$ and $\mu_2$ on $\mathbb{R}$ and $\RB$, respectively. Thus, in the following the crucial point will be to fix the measures $\mu_1, \mu_2$ and the parameter $t$.
\begin{remark}[Borel Measures]
\label{rem:BorelsRadons}
\begin{enumerate}
\item
\label{rem:BorelsRadons1}
If, in the situation of Lemma \ref{lemma:Radon}, $\mu_\RR$ and $\mu_\B$ are Borel measures and $\mu_\RR$ is finite, the sum $\mu_\RR\oplus \mu_{\B}$ is still a Borel measure. Thus, by the same arguments as in the above proof, a Radon measure if $\mu_\RR$ and $\mu_\B$ are so.

In fact, since $\mu_\RR$ is finite, for $\x\in \RR\subseteq \qR$ the open set $\RR$ is a neighbourhood of $\x$ in $\qR$ with $(\mu_\RR\oplus \mu_{\B})(\RR)<\infty$. Similarly, since $\mu_\B$ is locally finite, for $\x\in \RB\subseteq \qR$ we find some open neighbourhood $V$ of $\x$ in $\RB$:
\begin{align*}
	V=\widehat{f}_1^{-1}(U_1)\cap{\dots}\cap \widehat{f}_{k}^{-1}(U_k)\quad\text{with}\quad f_1,{\dots}, f_k \in \CAP\quad\text{and}\quad U_1,{\dots},U_k\subseteq \RR\:\text{ open}
\end{align*}
with $\mu_\B(V)<\infty$. Thus, since $\RR\sqcup V$ is open in $\qR$ and $\mu_\RR$ is finite, $(\mu_\RR\oplus \mu_{\B})(\RR\sqcup V)< \infty$ holds as well.
\item
\label{rem:BorelsRadons2}
If $\mu$ is a Borel (Radon) measure but not finite, the restrictions $\mu|_{\Borel(\mathbb{R})}$ and $\mu|_{\Borel(\RB)}$ are still Borel (Radon) measures because $\mu$ is locally finite and the topology on $\qR$ induces the usual topologies on $\RR$ and $\RB$, respectively. 
\end{enumerate}
\end{remark}
Combining the two parts of the above remark, we see that each (finite) Borel measures on $\qR$ is (exactly) of the form $\mu_\RR \oplus \mu_\B$ for (finite) Borel measures $\mu_\RR$ and $\mu_\B$ on $\RR$ and $\RB$, respectively. 
This provides us with the following statement (see also Updates in \cite{ChrisSymmLQG}) concerning the assumptions made in \cite{Jon} on the inner product on (the image under the Gelfand transform of) $C_0(\RR)\oplus \CAP$.
\begin{lemma}
	\label{lemma:JonsmitRadon}
  Let $\mu=\mu_\RR \oplus \mu_\B$ be a Borel measure on $\qR$ with
  \begingroup
\setlength{\leftmargini}{17pt}
\begin{itemize}
\item
\itspacecc
$\big\langle \widehat{f}_0,\widehat{f}_\AP\big\rangle_{\xi(\mu)}=0$\hspace{115.5pt}$\forall\: f_0\in C_0(\RR),\:\:\forall\:f_\AP\in \CAP$ 
\item
\itspacecc
$\big\langle \widehat{f}_\AP,\widehat{g}_\AP\big\rangle_{\xi(\mu)}=\big\langle \widehat{f}_\AP,\widehat{g}_\AP\big\rangle_{\muB}$\hspace{40.45pt}$\forall\: f_\AP, g_\AP\in \CAP$,
\end{itemize}
\endgroup
\noindent
whereby on the right hand side of the second point the Gelfand transforms have to be understood as elements of $C(\RB)$. 
 Then, 
   \begingroup
   \setlength{\leftmargini}{17pt}
 \begin{itemize}
 	\item
 	\itspacecc
 	$\mu_\RR=0$
 	\item
 	\itspacecc
 	$\mu_\B$ is inner regular\qquad $\Longrightarrow$\qquad  $\mu_\B=\muB$.
 \end{itemize}
 \endgroup
 \noindent
 In particular, $0_\RR \oplus \muB$ is the only Radon measure on $\qR$ which fulfils the above conditions. 
\end{lemma}
\begin{proof}
 	For $f_\AP=1$, the first condition reads
  \begin{align*}
  0= \int_{\xi(\qR)} \widehat{f_0}\:\dd\xi(\mu)&= \int_{\qR} \widehat{f_0}\cp \xi \:\dd\mu
  =\int_{\RR} \widehat{f_0}\cp \xi\:\dd\mu_\RR + \int_{\RB} \widehat{f_0}\cp \xi\:\dd\mu_{\B}
  \stackrel{\eqref{eq:Ksiii}}{=}\int_{\RR} f_0 \:\dd\mu_\RR
  \end{align*}
  for all $f_0\in C_0(\RR)$. Since $\mu_\RR(\RR)=\lim_n \mu_\RR([-n,n])$, and since for each $n\in \NN_{>0}$ we find a positive function in $C_0(\RR)$ which is identically $1$ on $[-n,n]$, we conclude that $\mu_\RR=0$ holds.
  
  Then, evaluating the second condition for $g_\AP=1$, we obtain 
   \begin{align*}
  \int_{\RB} \widehat{f}_\AP\cp \xi\:\dd\mu_{\B}=\int_{\RB} \widehat{f}_\AP\:\dd\muB\qquad\forall\: f_\AP\in \CAP.
  \end{align*}         
  Since by \eqref{eq:Ksiii} the left hand side equals the integral of $\widehat{f}_\AP$ (understood as an element of $C(\RB)$) over $\RB$ w.r.t.\ $\mu_{\B}$, Riesz-Markov shows that $\mu_{\B}=\muB$ holds. As it is easily checked that the measure $0_\RR \oplus \muB$ fulfils the above conditions, the claim follows.  
\end{proof}

\subsection{Motivation of the Construction}
\label{subsec:Motivation}
In this subsection, we will investigate whether the projections maps used for the definition of the Ashtekar-Lewandowski measure on $\A$  are appropriate for defining a reasonable measure on $\qR$ as well. This will lead us to the constructions of Subsection \ref{subsec:ProjStrucon}.

So, we start with the same maps $\pi_\alpha$ being also used in Appendix \ref{sec:AshLewMeasure} to define the Ashtekar-Lewandowski measure $\mAL$ on $\ovl{\Con}$. More precisely, we consider the projection maps 
\begin{align}
\label{eq:projdef}
\begin{split}
  \pi_\alpha \colon \qR&\rightarrow \SU^{|\alpha|}\\
  \ovl{x}&\mapsto (\Delta(\ovl{x})(\gamma_1),{\dots},\Delta(\ovl{x})(\gamma_{|\alpha|})) 
\end{split}
\end{align}
for $\alpha =(\gamma_1,{\dots},\gamma_k)\in \Gamma :=\bigsqcup_{n=1}^\infty \Paths^n$, $|\alpha|:=k$ 
and $\Delta:=\homiso \cp \oishomc\cp\xi\colon \qR \rightarrow \Homm_\inv(\Paths,\SU)$
\begin{align}
  \label{eq:commdiag}
  \begin{split}
  \makebox[0pt]{
    \begin{xy}
      \xymatrix{
        \Delta\colon\hspace{-35pt} &\qR \ar@{->}[r]^-{\xi}_-{\cong}& \Spec(\ishomc^*\cC_\w)\ar@{->}[r]^-{\oishomc}_-{\cong}& \ovl{\ConG} & \hspace{-26pt}\subseteq\AQR \ar@{->}[r]^-{\homiso}_-{\cong}& \Homm_\inv(\Paths,\SU)  
        \\
        &&&  \RR \ar@{->}[r]^-{\ishomc}  \ar@{->}[u]^-{\iota_{\RR}}  & \ConG \ar@{->}[u]^-{\iota_{\Con}}& \hspace{-107pt}\subseteq\Con. &  
      }
    \end{xy}
  }
  \end{split}
\end{align}
As in Appendix \ref{sec:AshLewMeasure}, we now have to choose a subset $\gred$ of $\Gamma$ which can be turned into a directed set. Here, we have to take the following issues into account: 
\begingroup
\setlength{\leftmargini}{14pt}
\begin{itemize}
\item
  \itspacecc
  $\gred$ has to be large enough to guarantee that an element $\ovl{x}\in \qR$ is completely determined by the values $\pi_\gamma(\ovl{x})$ for $\gamma\in \gred$. 
\item
  \itspacec
  Since the elements of $\qR$ correspond to invariant homomorphisms, for each $\x\in \qR$ the values $\pi_{\gamma_1}(\x)$ and $\pi_{\gamma_2}(\x)$ are related via \eqref{eq:InvGenConnRel} if the curves $\gamma_1$ and $\gamma_2$ only differ by an euclidean transformation. In particular, in view of defining measures on the images of the maps $\pi_{(\gamma_1,\gamma_2)}$, this will give further restrictions to the set $\gred$. Indeed, for $\gamma_1,\gamma_2$ linear and related as above, this image is given by some 1-dimensional submanifold. In contrast to that, if $\gamma_1$ and $\gamma_2$ are not related, it will be homeomorphic to $\TO^2$, see also discussions below. 
\item
  \itspacec
  For each $\alpha\in \gred$ we have to find some reasonable measure $\mu_\alpha$ on $\im[\pi_\alpha]\subseteq \SU^{|\alpha|}$. Here, it can be seen from \eqref{eq:meas0} that we cannot stick to the Haar measures on $\SU^{|\alpha|}$ if we want to obtain something that is different from the zero measure on $\qR$.  
\end{itemize}
\endgroup
\noindent
Concerning the first point, we recall that the $\Cstar$-algebras $\ishomc^*\cC_{\w}$ and $\ishomc^*\cC_{\lin\mc}$ coincide \cite{ChrisSymmLQG}, i.e., that $\qR\cong \Spec(\ishomc^*\cC_{\lin\mc})$ holds. Here, $\cC_{\lin\mc}$ denotes the $\Cstar$-algebra of cylindrical functions that corresponds to the set of curves $\Paths_{\lin\mc}:=\Paths_\lin \sqcup \Paths_\mc$ for $\Paths_\lin$ and $\Paths_\mc$ defined as follows:
\begingroup
\setlength{\leftmargini}{14pt}
\begin{itemize}
\item
  $\Paths_\lin$ denotes the set of linear curves of the form 
  $x+\gamma_{\vec{v},l}$ for $x, \vec{v}\in \mathbb{R}^3$ with $\|\vec{v}\|=1$ and $\gamma_{\vec{v},l}\colon [0,l]\rightarrow \mathbb{R}$, $t\mapsto t \cdot\vec{v}$. 
\item
  $\Paths_\mc$ denotes the set of circular curves of the form
  \begin{align*}
    \gc{n}{r}{x}{m} \colon [0,2\pi m]&\rightarrow \mathbb{R}^3\\
    t&\mapsto x + \cos(t)\cdot\vec{r} + \sin(t) \cdot\vec{n}\times \vec{r}  
  \end{align*}
  for $\vec{n},\vec{r},x\in \mathbb{R}^3$ with $\|\vec{n}\|=1$ and $0< m< 1$.
\end{itemize}
\endgroup
\vspace{-3pt}
\noindent
Thus, by \eqref{eq:pigammawquer} and the definition \eqref{embext} of $\oishomc$, it suffices to consider the curves in $\Paths_{\lin\mc}$ in order to satisfy condition \ref{def:ProjLim3}) from Definition \ref{def:ProjLim}. 
Moreover, due to invariance \eqref{eq:InvGenConnRel}, it suffices to consider the set $\pred$ of all linear and circular curves of the form 
\begin{align*}
	\gamma_l:=\gamma_{\vec{e}_1,l}\qquad\qquad\text{and}\qquad\qquad\gamma_{m,r}:=\gamma_{\vec{e}_3,r\hs\vec{e}_1}^{0,m}.
\end{align*}
Thus, in a first step, we end up with the set
\begin{align*}
  \gred:=\{(\gamma_1,{\dots},\gamma_k) \:|\: k\in \mathbb{N}_{>0}, \: \gamma_1,{\dots},\gamma_k \in \pred\}.
\end{align*}
Now, for each $\x\in \qR$ we have $\pi_{\gamma_l}(\x)\in \TO_{\vec{e}_1}$. This follows from the invariance condition \eqref{eq:InvGenConnRel} or directly from the fact that for each $\x\in \qR$ and $l\in \RR$ we find some $x_0\in \RR$ with $\xi(\x)(\chi_l)=\chi_l(x_0)$:
\begin{align}
  \label{eq:Linearbild}
  \begin{split}
    \pi_{\gamma_l}(\x)&\hspace{-1pt}=\xi(\x)\left(\mathrm{Re}(\chi_l)\right)\cdot \me - \xi(\x)\left(\mathrm{Im}(\chi_l)\right)\cdot \murs\left(\vec{e}_1\right)\\
    &=\cos(x_0\cdot l)\cdot \me - \sin(x_0\cdot l)\cdot \murs\left(\vec{e}_1\right)\\
    &=\exp(x_0\cdot l\cdot \murs(\vec{e}_1)).
  \end{split}
\end{align}
Indeed, the first equality follows from \eqref{eq:pigammawquer}  
and (see e.g. Subsection 4.3 in \cite{InvConnLQG})
\begin{align}
\label{eq:patralinear}
  \big(h_{\gamma_{l}}\cp \ishomc\big)(c)=\pr_2\cp \parall{\gamma_{l}}{\ishomc(c)}(0,\me)
  =\cos(c\hs l)\cdot\me - \sin(c\hs l)\cdot\murs\left(\vec{v}\right).
\end{align}
Since $\im[\pi_{\gamma_l}]=\TO_{\vec{e}_1}$ is a Lie subgroup of $\SU$ isomorphic to the circle $\TO$, it is natural to choose $\mu_{\alpha}$ to be the Haar measure on $\TO^{k}_{\vec{e}_1}\cong \TO^k$ if $\alpha$ is of the form $(\gamma_{l_1},{\dots},\gamma_{l_k})$. Moreover, in order to guarantee that $\mu_\alpha(\im[\pi_\alpha])\neq 0$ holds, we might only allow indices $(\gamma_{l_1},{\dots},\gamma_{l_k})$ for which $l_1,{\dots} ,l_k$ are $\mathbb{Z}$-independent. In fact, in this case it follows from  
Kronecker's theorem (cf. Theorem 4.13 in \cite{Broecker1985}) 
that $\pi_\alpha(\RR)$ is dense in $\TO^{k}_{\vec{e}_1}$, hence $\TO^{k}_{\vec{e}_1}= \ovl{\im[\pi_\alpha]}=\pi_\alpha\Rq$  
by compactness of $\qR$ and denseness of $\RR$ in $\qR$.\footnote{Alternatively, one can effort Lemma \ref{prop:OpenMapp}.\ref{lemma:OpenMapp4}.} 
On the level of linear curves, this is the same as to consider the directed set\footnote{Define $(l_1,{\dots},l_k)\leqZ \big(l'_1,{\dots},l'_{k'}\big)\:\: \Longleftrightarrow\:\: l_i\in \spann_{\mathbb{Z}}\big(l'_1,{\dots},l'_{k'}\big)
  \text{ for all } 1\leq i\leq k$. Then, $(I,\leqZ)$ is directed because $\mathbb{R}$ is a $\mathbb{Q}$ vector space and $l_1,{\dots},l_k\in \mathbb{R}$ are $\mathbb{Z}$-independent iff they are $\mathbb{Q}$-independent.}  
\begin{align}
  \label{eq:dirset}
  I:=\{(l_1,{\dots},l_k)\in \RR^k \:|\: k\in \mathbb{N}_{>0}, \: l_1,{\dots},l_k \text{ are } \mathbb{Z}\text{-independent}\},
\end{align}
the projection maps $\pi_{L}\colon \qR \rightarrow \TO^{|L|}$,\:\: $\ovl{x}\mapsto (\xi(\ovl{x})(\chi_{l_1}),{\dots},\xi(\ovl{x})(\chi_{l_k}))$ for $L=(l_1,{\dots},l_k)\in I$, $|L|:=k$, and to take the Haar measure $\mu_{|L|}$ on the $k$-torus $\TO^k$. Then, if we restrict to $\RB \subseteq \qR$ and define the transition maps
$\pi^{L'}_{L}\colon \TO^{|L'|}\rightarrow \TO^{|L|}$ by
\begin{align*}
  \pi^{L'}_{L} (s_1,{\dots},s_{k'}):= \left(\prod_{i=1}^{k'}{s_i}^{n^i_1} ,{\dots}, \prod_{i=1}^{k'}{s_i}^{n^i_k} \right) 
  \qquad\text{ if }\qquad l_j=\sum_{i=1}^{k'} n^i_j\cdot l'_i 
\end{align*}
with $n_j^i\in \mathbb{Z}$ for $1\leq j\leq k$, $1\leq i\leq k'$, 
we obtain a projective structure and a consistent family $\{\mu_{|L|}\}_{L\in I}$ of normalized Radon measures that reproduce the Haar measure on $\RB$, cf. Section 4  in \cite{Vel}.
\newline
\vspace{-10pt}
\newline
However, we also have to take circular curves into account, and the first step towards this is to investigate the image of the maps $\pi_{\gamma_{m,r}}$. For this, recall that (see, e.g., Subsection 4.3 in \cite{InvConnLQG})
\begin{align}
  \label{eq:parallcirc}
  \pr_2\cp \parall{\gc{n}{r}{x}{m}}{\ishomc(c)}(\vec{r},e)=\exp\left(\textstyle\frac{\tau}{2}\cdot\murs(\vec{n})\right)\cdot \Co{\sigma}(A(\tau,c))
\end{align}       
for  $\sigma \in \SU$ with $\uberll{\sigma}{\vec{e}_3}=\vec{n}$ as well as 
\begin{align}
  \label{eq:matrixentr}
  \begin{split}
  A(\tau,c):=&\begin{pmatrix} \cos(\pc \mm)+\frac{i}{2\pc}\sin(\pc \mm) &\frac{cr}{\pc} \sin(\pc \mm)  \\ -\frac{cr}{\pc}\sin(\pc \mm) & \cos(\pc \mm)-\frac{i}{2\pc}\sin(\pc \mm)  \end{pmatrix}\\[4pt]
  =&\exp(-\textstyle\frac{\tau}{2}\cdot [2rc\cdot \tau_2+\tau_3]).
\end{split}
\end{align}
Here, $\pc:= \sqrt{c^2r^2+\frac{1}{4}}$ and $\Co{\sigma}$ denotes the conjugation by $\sigma$ in $\SU$. For the second equality in \eqref{eq:matrixentr} recall that
\begin{align}
  \label{eq:expformel}
  \exp\left(t \cdot \murs(\vec{n})\right)=\cos(t)\cdot\me + \sin(t)\cdot\murs(\vec{n})\qquad \forall\:\vec{n}\in \RR^3,\:\forall t\in \RR.
\end{align} 
The next lemma describes the images of the maps $\pi_{\delta}$ for $\delta \in \Paths_{\mc}$. 
\begin{lemma}
  \label{lemma:BildCirc}
  Let $\delta=\gc{n}{r}{x}{\mm}\in \Paths_{\mc}$ be fixed and $d:=\exp\left(\textstyle\frac{\tau}{2}\cdot\murs(\vec{n})\right)$. Then, 
  \begin{enumerate}
   \item
    \label{lemma:BildCirc1}	
    $\pi_\delta\Rq$ is of measure zero w.r.t.\ the Haar measure on $\SU$.
  \item
  	\label{lemma:BildCirc2}	
    There is no proper Lie subgroup $H\subsetneq \SU$ which contains $\pi_\delta\Rq$. 
  \item
    \label{lemma:BildCirc3}	
    The maps $\pi_\delta, \{\pi_{\gamma_l}\}_{l\in \mathbb{R}_{>0}}$  separate the points in $\qR$. 
  \item 
    \label{lemma:BildCirc4}	    
    We have  $\pi_\delta\Rq=\pi_\delta(\mathbb{R})\:\cup\: \underbrace{d \cdot T_{\dot\delta(0)}}_{\pi_\delta(\RB)}
    \qquad \text{with}\:\:\qquad \pi_\delta(\mathbb{R})\:\cap\: d\cdot T_{\dot\delta(0)}= \{\pm d\}$.
  \item
    \label{lemma:BildCirc5}	
     Let $a_n:=\frac{\sign(n)}{r}\sqrt{\frac{n^2\pi^2}{\tau^2}-\frac{1}{4}}$\: for\: $n\in \mathbb{Z}_{\neq 0}$ and
     \begin{align*}
     A_0:=(a_{-1},a_1)\qquad A_n:=(a_n,a_{n+1})\:\:\text{ for }\:\: n\geq 1\qquad A_n:=(a_{n-1}, a_n)\:\:\text{ for }\:\: n\leq -1.
     \end{align*}    
     Then, $\pi_{\gamma_{\tau,r}}|_{A_n}$ is injective for all $n\in \ZZ$, 
     \begin{align*}
     \pi_{\gamma_{\tau,r}}(a_{n})=d\:\:\text{ iff }\:\:|n|\:\:\text{ is even}\qquad\qquad \pi'_{\tau,r}(a_{n})=-d\:\:\text{ iff }\:\:|n|\:\:\text{ is odd}
     \end{align*}
     and $\pi_{\gamma_{\tau,r}}\big(A_m\big)\cap \pi_{\gamma_{\tau,r}}\big(A_n\big)=\emptyset$\: for all $m,n\in \ZZ$\: with\: $m\neq n$. 
     For increasing $|n|$, the sets
     \begin{align*}
     B_n:=[a_{2n},a_{2(n+1)}]\:\text{ for }\:n\geq 1\qquad\qquad B_n:=[a_{2(n-1)}, a_{2n}]\:\text{ for }\:n\leq -1
     \end{align*}  
     merge to $\TO_{\dot{\delta}(0)}$ in the following sense. For each 
     $\epsilon > 0$, we find $n_\epsilon\in \mathbb{N}_{\geq1}$ such that for $|n| \geq n_\epsilon$ we have
     \begin{align*}
     \forall\: s\in B_n:\exists\: s' \in \TO_{\dot{\delta}(0)}:\|s-s'\|_{\mathrm{op}}\leq \epsilon.
     \end{align*}
  \end{enumerate}
  \begin{proof}
    See Appendix \ref{sec:ProofOfLemmaImCirc}. 
  \end{proof}
\end{lemma}
The first and the second part of the above lemma already show that it is hard to equip $\im[\pi_{\gamma_{\tau,r}}]$ with a reasonable measure. 
In addition to that, it is difficult to define a reasonable ordering, i.e., a directed set labelling the projection spaces.  
Indeed, the first thing one might try, is to define $\gamma_{\tau,r}\leq \gamma_l$ or $\gamma_l\leq \gamma_{\tau,r}$. But, then one has to define reasonable transition maps between $\im[\pi_{\gamma_{\tau,r}}]$ and $\im[\pi_{\gamma_l}]$, i.e., maps 
\begin{align*}
\TT_1\colon \im[\pi_{\gamma_{\tau,r}}]\rightarrow \im[\pi_{\gamma_l}]  \quad\qquad &\text{or}\quad \qquad \TT_2\colon \im[\pi_{\gamma_l}]\rightarrow \im[\pi_{\gamma_{\tau,r}}]\\
\text{with} \hspace{80pt} \TT_1\cp \pi_{\gamma_{\tau,r}}=\pi_{\gamma_l}\quad\qquad &\text{or}\quad \qquad \TT_2\cp \pi_{\gamma_l} = \pi_{\gamma_{\tau,r}}, \hspace{75pt}\text{respectively}.
\end{align*}
 This, however, is difficult, and even impossible if $\tau r =l$ holds: 
\begingroup
\setlength{\leftmargini}{14pt}
\begin{itemize}
\item
  Let $0, \frac{2\pi}{l}\in \RR\subseteq \qR$, i.e., $0=\xi^{-1}(\iota_\RR(0))$ and $\frac{2\pi}{l}=\xi^{-1}\big(\iota_\RR\big(\frac{2\pi}{l}\big)\big)$, cf.\ \eqref{eq:iotaxiR}. Then
	\begin{align*}  
  	\pi_{\gamma_l}(0)=\pi_{\gamma_l}(\iota_\RR(0))\stackrel{\eqref{eq:Linearbild}}{=}\me=\pi_{\gamma_l}\big(\textstyle\frac{2\pi}{l}\big)\qquad\quad\text{but}\qquad\quad \pi_{\gamma_{\tau,r}}(0)\stackrel{\eqref{eq:parallcirc}}{=}\me\neq \pi_{\gamma_{\tau,r}}\big(\frac{2\pi}{l}\big).
	\end{align*}  	 
   Here, the inequality on the right hand side is because by Lemma \ref{lemma:BildCirc}.\ref{lemma:BildCirc4} we have that $\pi_{\gamma_{\tau,r}}(x)=\pi_{\gamma_{\tau,r}}(y)$ for $x\neq y$ enforces $\pi_{\gamma_{\tau,r}}(x)=\pm \exp(\frac{\tau}{2}\cdot \tau_3)$. Thus, there cannot exist a transition map $\TT_2\colon \im[\pi_{\gamma_l}]\rightarrow \im[\pi_{\gamma_{\tau,r}}]$, as then 
	\begin{align*}   
   	\pi_{\gamma_{\tau,r}}(0)=(\TT_2\cp \pi_{\gamma_l})(0)=(\TT_2\cp \pi_{\gamma_l})\big(\textstyle\frac{2\pi}{l}\big)=\pi_{\gamma_{\tau,r}}\big(\textstyle\frac{2\pi}{l}\big)
   \end{align*}
   would hold.
\item
  We have $\pi_{\gamma_{\tau,r}}(a_{2n})=d$ for all $n\in \mathbb{Z}_{\neq 0}$, so that for a transition map $\TT_1\colon \im[\pi_{\gamma_{\tau,r}}] \rightarrow \im[\pi_{l}]$ we would have  
  \begin{align*}
    \TT_1(d)=\left(\TT_1\cp\pi_{\gamma_{\tau,r}}\right)(a_{2n})=\pi_{\gamma_l}( a_{2n}) \qquad\forall \: n\in \ZZ_{\neq 0}.
  \end{align*}
  Then, for $\epsilon >0$ we find $n_\epsilon \in \mathbb{N}_{\geq 0}$ such that\footnote{It is clear that $\lim_n [l a_{2n}- 2\pi n] =0$, and that $\lim_n [l a_{2n}-  l a_{2(n+1)}]= 2\pi$. Hence, we have to show that we find $n_0\in \NN_{\geq 1}$ such that $l a_{2n}-  l a_{2(n+1)} \neq 2\pi$ for all $n\geq n_0$. Now, $l a_{2n}-  l a_{2(n+1)}=2n\pi\left[\sqrt{1-\textstyle\frac{\tau^2}{4(2n+2)^2\pi^2}}-\sqrt{1-\textstyle\frac{\tau^2}{4(2n)^2\pi^2}}\right] + 2\pi \sqrt{1-\textstyle\frac{\tau^2}{4(2n+2)^2\pi^2}}$, where the first summand is negative and tends to zero for $n\rightarrow \infty$. Since the second summand is smaller than $2\pi$, the whole expression is smaller than $2\pi$ for $n$ suitable large.} $l a_{2n}-  l a_{2(n+1)}\in B_\epsilon(2\pi)\backslash\{2\pi\}$ for all $n\geq n_\epsilon$. But, $\pi_{l}(a_{2n})=\pi_{l}(\raisebox{0pt}{$a_{2(n+1)}$})$ implies $la_{2(n+1)}-la_{2n} =k_n2\pi$ for some $k_n \in \mathbb{Z}$,
  so that we get a contradiction if we choose $\epsilon< 2\pi$.
\end{itemize}
\endgroup 
\noindent
Now, in order to satisfy Condition \ref{def:ProjLim3}) from Definition \ref{def:ProjLim}, by Lemma \ref{lemma:BildCirc}.\ref{lemma:BildCirc3} it suffices to 
take one fixed circular curve $\gamma_{\tau,r}$ into account. So, 
we can circumvent the above transition map problem by sticking to the directed set $I$ from \eqref{eq:dirset}. More precisely, we can incorporate the map $\pi_{\gamma_{\tau,r}}$ into each of the projection maps as follows: 
\begingroup
\setlength{\leftmargini}{14pt}
\begin{itemize}
\item
\label{it:fff}
  For $I\ni L=(l_1,{\dots},l_k)$ we can define 
	\begin{align*}
	\pi'_L\colon \qR &\rightarrow \SU^{k+1}\\
		\x&\mapsto \big(\pi_{\gamma_{l_1}}(\x),{\dots},\pi_{\gamma_{l_k}}(\x),\pi_{\gamma_{\tau,r}}(\x)\big).
	\end{align*}  
  But then $\im[\pi'_L]$ crucially depends on the $\mathbb{Z}$-independence of $l_1,{\dots},l_k,r\tau$ as, e.g., we have $\pi_L(\RB)=\TO_{\vec{e}_1}^k \times  \TO_{\vec{e}_2}$ in the $\mathbb{Z}$-independent case as well as $\pi_L(\RB)\cong \TO$ if $l_1,{\dots},l_k,r\tau$ are multiples of the same real number. For this, observe that we cannot restrict to independent tuples without adapting $\leq$. This is because for $(l',r\tau)$, $(l'',r\tau)$ independent and $(l_1,{\dots},l_k)\in I$ an upper bound of $L':=l',L'':=l''$, the tuple $(l_1,{\dots},l_k,r\tau)$
  does not need to be independent as well. In fact, for $l'=l-r\tau$ this cannot be true for any such $L$. 
  All this makes it difficult to find transition maps and suitable consistent families of normalized Radon measures on these spaces.
\item
\label{it:ffff}
  Basically, Lemma \ref{lemma:BildCirc}.\ref{lemma:BildCirc3} is due to the fact that the $C_0(\RR)$-part 
  of the function 
  $a\colon \RR\ni c\mapsto (\pi_{\gamma_{\tau,r}}(c))_{11}$  
  vanishes nowhere. Now, we can try to find some analytic curve $\gamma$ such that for one of the entries $(\pi_{\gamma}(\cdot))_{ij}$, $1\leq i,j\leq 2$, the $\CAP$-part is zero and the $C_0(\RR)$-part vanishes nowhere and is injective. Then, Condition \ref{def:ProjLim3}) from Definition \ref{def:ProjLim} would hold for the projection maps $\wt{\pi}_L\colon \qR \rightarrow \im[\pi_\gamma] \sqcup \TO_{\vec{e}_1}^k$
  \begin{align}
    \label{eq:PiLMuster}
    \wt{\pi}_L(\ovl{x}) := 
    \begin{cases} 
      \pi_\gamma(\x)  & \mbox{if } \ovl{x}\in \mathbb{R}\\
      \pi_L(\x) &\mbox{if } \ovl{x}\in \RB,
    \end{cases} 
  \end{align}
  and we could define the transition maps and measures (once the Borel $\sigma$-algebra is shown to split up) on $\im[\pi_\gamma]$ and $\TO_{\vec{e}_1}^k$ separately. However, even if such a curve $\gamma$ exists, it is not to be expected that it is easier to find reasonable measures on $\im[\pi_\gamma]$ than on $\im[\pi_{\gamma_{\tau,r}}]$.  
\end{itemize}
 \endgroup 
 \noindent
In the next subsection, we will follow the philosophy of the second approach. More precisely, we will use distinguished generators of $C_0(\RR)\oplus \CAP$ in order to define projective structures on $\qR$ in a much more direct way. This will allow us to circumvent the image problem we have for the projection maps $\pi_{\gamma_{\tau,r}}$ and $\pi_\gamma$ discussed above. Thus, the crucial part will not be the definition of the projective structure on $\qR$ but to determine the respective consistent families of normalized Radon measures. Here, the main difficulties will arise from determining the Borel $\sigma$-algebras of the projection spaces.

\subsection{Projective Structures on $\boldsymbol{\qR}$}
\label{subsec:ProjStrucon}
In this subsection, we will construct projective structures on $\qR$ by means of generators of $C_0(\RR)\oplus\CAP$ which are of the form $\{f\}\sqcup \{\chi_l\}_{l\in \mathbb{R}}$ with $f_0\in C_0(\RR)$ nowhere vanishing and injective, see Lemma \ref{lemma:WeierErzeuger}.\ref{lemma:WeierErzeuger1}. In fact, due to injectivity, $f$ is a homeomorphism (Lemma \ref{lemma:WeierErzeuger}.\ref{lemma:WeierErzeuger2}) so that the projection maps \eqref{eq:PiLMuster} can directly be translated to the niveau of the generators  $\{f\}\sqcup \{\chi_l\}_{l\in \mathbb{R}}$. This will be done in the definition following
\begin{lemma}
  \label{lemma:WeierErzeuger}
  Let $f\in C_0(\RR)$ vanish nowhere.
  \begin{enumerate}
  \item
    \label{lemma:WeierErzeuger1}
    The functions $\{f\}\sqcup \{\chi_l\}_{l\in \mathbb{R}}$ generate a dense $^*$-subalgebra of $C_0(\mathbb{R})\oplus \CAP$. If $f$ is injective, it generates a dense $^*$-subalgebra of $C_0(\RR)$.   
  \item
    \label{lemma:WeierErzeuger2}
    If $f$ is injective, it is a homeomorphism. 
  \end{enumerate}
  \begin{proof}
    \begin{enumerate}
    \item
	  Since $f(x)\neq 0$ for all $x\in \mathbb{R}$, the $^*$-algebra generated by $\{f \cdot\chi_l\}_{l \in \mathbb{R}}\subseteq C_0(\mathbb{R})$ separates the points in $\mathbb{R}$ and vanishes nowhere. Consequently, it is dense in $C_0(\mathbb{R})$ by the complex Stone-Weierstrass theorem for locally compact Hausdorff spaces. Since $\CAP$ is generated by the functions $\{\chi_l\}_{l \in \mathbb{R}}$, the first claim follows. If $f$ is in injective, the $^*$-algebra generated by $f$ is dense in $C_0(\mathbb{R})$ because $f$ separates the points in $\mathbb{R}$ and vanishes nowhere.     
    \item
   	Let $\RR\sqcup \{\infty\}$ denote the one point compactification of $\RR$. Then, $\ovl{f}\colon \RR\sqcup \{\infty\}\rightarrow \CCC$ defined by $\ovl{f}(\infty):=0$ and $\ovl{f}|_\RR:=f$ is continuous and injective, hence a homeomorphism onto its image $\im[f]\sqcup\{0\}$. Thus, $f^{-1}=\ovl{f}^{-1}|_{\im[f]}$ is continuous as well.
    \end{enumerate}
  \end{proof}
\end{lemma}
The next definition basically adapts (and collects) some of the notations we have already introduced in the previous subsection.
\begin{definition}
  \label{def:ProjLimit}
  Assume that $f\in C_0(\mathbb{R})$ is injective and $f(x)\neq 0$ for all $x\in \mathbb{R}$.
  \begin{enumerate}
  \item 
  \label{def:ProjLim1}
    Let $I$ denote the set of all finite tuples $L=(l_1,{\dots},l_k)$ consisting of $\mathbb{Z}$-independent real numbers $l_1,{\dots},l_k$, and let $|L|$ denote the length $k$ of the tuple $L$.
  \item
  \label{def:ProjLim2}
    For $L,L'\in I$ define $L\leqZ L'$ iff $l_i \in \spann_{\mathbb{Z}}(l'_1,{\dots},l'_{k'})$ for all $1\leq i\leq k$. 
  \item
  \label{def:ProjLim3}
    For $L\in I$ and $k:=|L|$ define $\pi_L\colon \qR \rightarrow \prfl{f}{L}=: X_L$ by
    \begin{equation}
      \label{eq:PiL}
      \pi_L(\ovl{x}) := 
      \begin{cases} 
        f(\ovl{x})  & \mbox{if }\:\:\: \ovl{x}\in \mathbb{R}\\
        (\ovl{x}(\chi_{l_1}),{\dots},\ovl{x}(\chi_{l_k})) &\mbox{if }\:\:\: \ovl{x}\in \RB,
      \end{cases} 
    \end{equation}
    and equip $X_L$ with the final topology w.r.t.\ this map.
  \item
    For $L,L'\in I$ with $L\leqZ L'$ define $\pi_L^{L'}\colon X_{L'}\rightarrow X_L$ by $\pi_L^{L'}(y):=y$ for $y\in \im[f]$ as well as
    \begin{align}
      \label{eq:transit}
      \pi^{L'}_{L} (s_1,{\dots},s_{k'}):= \left(\prod_{i=1}^{k'}{s_i}^{n^i_1} ,{\dots}, \prod_{i=1}^{k'}{s_i}^{n^i_{j}} \right)
      \qquad\text{ if }\qquad l_j=\sum_{i=1}^{k'} n^i_jl'_i
    \end{align}
    with $n_j^i\in \mathbb{Z}$ for $1\leq j\leq k=|L|$, $1\leq i\leq k'=|L'|$ and $(s_1,{\dots}, s_{k'})\in \TO^{|L'|}$.
  \end{enumerate} 
\end{definition}
We now show 
that $\qR$ is indeed a projective limit of $\{X_L\}_{L\in I}$. Moreover, we determine the Borel $\sigma$-algebras of the spaces $X_L$. This will lead to an analogous decomposition of finite Radon measures on the $X_L$ as we have for the space $\qR$. Here, the crucial point will be to show that the subspace topologies of $\im[f]$ and $\TO^{|L|}$ w.r.t.\ the final topology on $X_L$ are just their canonical ones.
For this, we will need the following definitions and facts:
\begingroup
\setlength{\leftmargini}{14pt}
\begin{itemize}
\item
  Let $\T_f$ and $\T_{L}$ denote the standard topologies
  on $\im[f]$ and $\TO^{|L|}$, respectively, i.e., the subspace topology on $\im[f]$ inherited from $\mathbb{R}$ and product topology on $\TO^{|L|}$.
\item
  \itspacecc 
  For $L\in I$ let $\pih_L \colon \RB \rightarrow \TO^{|L|}$ denote the restriction of $\pi_L$ to $\RB$. 
\item    
  \itspacecc
  Analogously, let $\pih_L^{L'} \colon \TO^{|L'|} \rightarrow \TO^{|L|}$ denote the restriction of $\pi^{L'}_{L}$ to $\TO^{|L|}$ if $L,L'\in I$ with $L\leqZ L'$.
\item
  \itspacecc
   For $q\in \Q_{\neq 0}$ define $\chi_{l,q}:=\chi_{l\slash q}$.
\item
  \itspacecc
  Since each $L\in I$ consists of 
  $\mathbb{Q}$-independent reals, we find (and fix) a subset $L^\perp\subseteq \mathbb{R}$ such that $\ovl{L}:=L\sqcup L^\perp$ is a $\mathbb{Q}$-base of $\mathbb{R}$. It is clear that, together with the constant function $1=\chi_{0}$, the functions $\left\{\chi_{l,n}\right\}_{(l,n) \in \ovl{L} \times \NN_{>0}}$ generate a dense $^*$-subalgebra 
  of $\CAP$.   
\item 
  \itspacecc
  For $p\in \mathbb{N}_{\geq 1}$ and $A\subseteq \TO$ define $\hat{p}\colon \TO\rightarrow \TO$,\:\: $s\mapsto s^p$ and let
  \begin{align*}
    \sqrt[p]{A}:=\{s\in T\:|\: s^p\in A\}\qquad\text{as well as}\qquad A^p:=\{s^p\:|\: s\in A\}.
  \end{align*}   	
  If $\OO\subseteq \TO$ is open, then $\OO^p$ and $\sqrt[p]{\OO}=\hat{p}^{-1}(\OO)$ are open as well. This is because $\hat{p}$ is open (inverse function theorem) and continuous.  
\item
  \itspacecc
  For $A\subseteq \TO$ and $m\in \ZN$, we define
  \begin{equation*}
    A^{\sign(m)}:=
    \begin{cases}
      A & \mbox{if }m >0\\
      \{\ovl{z}\:|\: z\in A\} &\mbox{if }m <0.
    \end{cases}
  \end{equation*}
\end{itemize}
\endgroup
\noindent
The next lemma highlights the relevant properties of the maps $\pih_L$.
\begin{lemma}
  \label{prop:OpenMapp}
  Let $L=(l_1,{\dots},l_k)\in I$.
  \begin{enumerate}
  \item 
  \label{lemma:OpenMapp1}
    Let $\psi \in \RB$, $q_i\in \Q$ and $s_i\in \TO$ for $1\leq i\leq k$. Then, we find $\psi'\in \RB$ with 
    \begin{align*}    
    \psi'(\chi_{l_i,q_i})=s_i\quad  \forall\:1\leq i\leq k\qquad\quad\text{and}\qquad\quad \psi'(\chi_{l})=\psi(\chi_{l})\quad\forall\: l \in \spann_{\mathbb{Q}}(L^\perp). 
    \end{align*}
  \item
  \label{lemma:OpenMapp2}
    Let $l\in \RR$, $m_i\in \ZN$ and $\OO_i\subseteq \TO$ open for $1\leq i\leq n$. Then, for $m:=|m_1\cdot{\cdots} \cdot m_n|$ and $p_i:=|\frac{m}{m_i}|$ we have
    \begin{align}
      \label{eq:WurzelUrbild}
      \bigcap_{i=1}^n \chih^{-1}_{l,m_i}(\OO_i)=\bigcap_{i=1}^n \chih^{-1}_{l,m}\left(\left[\!\sqrt[p_i]{\OO_i}\:\right]^{\sign(m_i)}\right)
      =\chih^{-1}_{l,m}\left(\OO\right).
    \end{align} 
    for the open subset $\OO=\displaystyle\left[\!\sqrt[p_1]{\OO_1}\:\right]^{\sign(m_1)} \cap {\dots} \cap \left[\!\sqrt[p_n]{\OO_n}\:\right]^{\sign(m_n)}\subseteq \TO$.
  \item  
  \label{lemma:OpenMapp3}
     Let $A_i, B_j\subseteq \TO^1$ for $1\leq i\leq k$, $1\leq j\leq q$, with $B_1,{\dots},B_q\neq \emptyset$. Moreover, let $h_1,{\dots},h_q \in \mathbb{R}$ such that $l_1,{\dots},l_k$, $h_1,{\dots},h_q$ are $\mathbb{Z}$-independent. For $m_1,{\dots},m_k, n_1,{\dots},n_q \in \ZN$ let
     \begin{align}
     \label{eq:BaseRbohr}
     \!\!\!\!W:=\underbrace{\chih_{l_1,m_1}^{-1}(A_1)\cap{\dots} \cap \chih_{l_k,m_k}^{-1}(A_k)}_{U}\:\cap\: \underbrace{\chih_{h_1,n_1}^{-1}(B_1)\cap {\dots} \cap \chih_{h_q,n_q}^{-1}(B_q)}_{U'}.
     \end{align}
     Then, $\pih_L(W)= A_1^{m_1} \times {\dots} \times A_k^{m_k}$ holds.
  \item
  \label{lemma:OpenMapp4}
    The map $\pih_L$ is surjective, continuous and open.
  \end{enumerate}
  \end{lemma}
  \begin{proof}
    \begin{enumerate}
    \item
        We choose $\{x_\alpha\}_{\alpha\in I}\subseteq \RR$ such that $\chi_{q_\alpha\cdot l_\alpha}(x_\alpha)=s_\alpha$ holds for all $\alpha\in I$, and define
        \begin{align*}
        \zeta(\chi_{q\cdot l_\alpha}):=\chi_{q\cdot l_\alpha}(x_\alpha)\: \text{ if }\:\alpha\in I \qquad\quad\text{as well as}\qquad\quad \zeta(\chi_{q\cdot l}):=\psi(\chi_{q\cdot l})\: \text{ if }\:l\in L^\perp
        \end{align*}
        for all $q\in \mathbb{Q}$. 
        Then, for $l\in \RR$ arbitrary, we have a unique representation of the form 
        \begin{align*}    
        l=\textstyle\sum_{i=1}^k q_i\: l_{\alpha_i} + \textstyle\sum_{j=1}^{k'} q_j'\: l'_j
        \end{align*} 
        with $l'_j \in L^\perp$ and $q_i,q'_j \in \Q$ for $1\leq i\leq k$, $1\leq j\leq k'$. Here, we define
        \begin{align*}
        \zeta\!\left(\chi_{l}\right):=\textstyle\prod_{i=1}^k\zeta\big(\raisebox{1pt}{$\chi_{q_i\cdot l_{\alpha_i}}$}\big) \cdot \textstyle\prod_{j=1}^q\zeta\big(\raisebox{1pt}{$\chi_{q'_j\cdot l'_j}$}\big) \qquad\quad\text{as well as}\qquad\quad \zeta(1):=1.
        \end{align*}
         It is easy to see that $\zeta'\colon \RRD\rightarrow \TO$ is a well-defined homomorphism which, due to the last part of Convention \ref{conv:Boundedfunc}, gives rise to a well-defined element $\psi'\in \RB$. By construction, $\psi'$ has the desired properties. 
    \item 
      Obviously, $\OO$ is open, and since
      the second equality in \eqref{eq:WurzelUrbild} is clear, it suffices to show that
      \begin{align*}
        \chih^{-1}_{l,m}(A)=\chih^{-1}_{l,p\cdot |m|}\left(\big[\!\raisebox{0ex}{$\sqrt[p]{A}$}\:\big]^{\sign(m)}\right)
      \end{align*}
      holds for $A\subseteq \TO$, $l\in \RR$, $p \in \mathbb{N}_{\geq 1}$ and $m\in \ZN$. 
      For the inclusion $\supseteq$, we let
      \begin{align*}
      \psi \in \chih^{-1}_{l,p\cdot |m|}\left(\big[\!\raisebox{-0.3ex}{$\sqrt[p]{A}$}\:\big]^{\sign(m)}\right).
	  \end{align*}           
       Then,  
      \begin{align*}
        \psi(\chi_{l,p\cdot |m|}) \in \big[\!\raisebox{-0.3ex}{$\sqrt[p]{A}$}\:\big]^{\sign(m)}\qquad\Longrightarrow\qquad \psi(\chi_{l, m})=\big[\raisebox{-0.15ex}{$\psi(\chi_{l,p\cdot |m|})^{p}$}\big]^{\sign(m)}\in A.
      \end{align*}
      For the converse inclusion, let $\psi \in\chih_{l,m}^{-1}(A)$. Then, 
      \begin{align*}
        \big[\raisebox{-0.2ex}{$\psi(\chi_{l,p\cdot |m|})^{p}$}\big]^{\sign(m)}=\psi(\chi_{l,m})\in A\qquad \Longrightarrow\qquad \psi(\chi_{l,p\cdot |m|})\in \big[\!\raisebox{-0.15ex}{$\sqrt[p]{A}$}\:\big]^{\sign(m)}.
      \end{align*}
    \item
      We proceed in two steps:    
      \begingroup
      \setlength{\leftmarginii}{15pt}
      \begin{itemize}
      	\item
      	\itspace
      	We show that $\pih_L(W)=\pih_L(U)$ holds. For this, it suffices to verify $\pih_L(U)\subseteq \pih_L(W)$ because the converse inclusion is clear from $W\subseteq U$. 
      	So, for $\psi \in U$ we have to show that $\pih_L(\psi) \in \pih_L(W)$. Since $B_j\neq \emptyset$, we find $z_j \in B_{j}$ 
      	for all $1\leq j\leq q$.
      	By \ref{lemma:OpenMapp1}), we find $\psi'\in \RB$ with $\psi'(\chi_{l_i,m_i})=\psi(\chi_{l_i,m_i})\in A_i$ for all $1\leq i\leq k$ and $\psi'(\chi_{h_j,n_j})=z_j\in B_j$ for all $1\leq j\leq q$. This shows $\psi'\in U\cap U'=W$, hence $\pih_L(\psi')\in \pih_L(W)$. Thus,  
      	\begin{align*}
      	\pih_L(\psi)&=(\psi(\chi_{l_1,m_1}),{\dots},\psi(\chi_{l_k,m_k}))
      	=\big(\psi'(\chi_{l_1,m_1}),{\dots},\psi'(\chi_{l_k,m_k})\big)=\pih_L(\psi')\in \pih_L(W).
      	\end{align*}
      	\item
      	We show that $\pih_L(U)=A^{m_1}_1 \times {\dots} \times A^{m_k}_k$ holds. For this, it suffices to verify the inclusion $\supseteq$ as the opposite inclusion is clear from the definitions. For this, let $s_i \in A_i^{m_i}$ for $1\leq i\leq k$ and $z_i\in A_i$ with $z_i^{m_i}=s_i$. Then,   
      	\ref{lemma:OpenMapp1}) provides us with some $\psi'\in \RB$ with $\psi'(\chi_{l_i,m_i})=z_i\in A_i$ for $1\leq i\leq k$. Thus, $\psi'\in U$ and $\psi'(\chi_{l_i})=z_i^{m_i}=s_i \in A_i^{m_i}$ for all $1\leq i\leq k$.
      \end{itemize}
      \endgroup
    \item
     Continuity of $\pih_L$ is clear from 
     \begin{align*}
     \pih_L^{-1}(A_1,{\dots},A_k)= \chih_{l_1}^{-1}(A_1)\cap{\dots} \cap \chih_{l_k}^{-1}(A_k)\qquad\forall\:A_1,{\dots} A_k \subseteq \TO,
     \end{align*}
     and surjectivity is clear from Part \ref{lemma:OpenMapp3}) if we choose $A_1,{\dots},A_k=\TO$.
     
     For openness observe that the $^*$-algebra generated by $1$ and $\{\chi_{l,m}\}_{(l,m) \in \ovl{L} \times \ZN}$ is dense in $\CAP$ as it equals the $^*$-algebra generated by all characters $\chi_l$. Thus, the subsets of the form $\chi_{l,m}^{-1}(\OO)$ with $\OO\subseteq \TO$ open and $(l,m) \in \ovl{L}\times \ZN$ 	
     provide a subbasis for the topology of $\RB$.\footnote{This is because the Gelfand topology on $\RB$ equals the initial topology w.r.t.\ the Gelfand transforms of the elements of each subset of $\CAP$ which generates a dense $*$-algebra of $\CAP$. Consequently, the Gelfand topology on $\RB$ equals the initial topology w.r.t.\ the functions $1$ and $\chi_{l,m}$ for $(l,m) \in \ovl{L}\times \ZN$. 
     	Since the preimage of a subset of $\mathbb{C}$ under $\widehat{1}$ is either empty or $\RB$, the statement follows.} 
     Consequently, a base of this topology  is given by all finite intersections of such subsets. Then, Part \ref{lemma:OpenMapp2}) shows that, in order to obtain a base for the topology of $\RB$, it suffices to consider intersections of the form \eqref{eq:BaseRbohr} with $A_1,{\dots},A_k,B_1,{\dots},B_q$ open in $\TO$. In fact, since $\chih^{-1}_{l,m}(\TO)=\RB$, we can assume that each $l_1,{\dots},l_k$ occurs in each of these intersections. 
     Thus, since $A_i^{m_i}$ is open if $A_i$ is open, Part \ref{lemma:OpenMapp3}) shows that $\pih_L$ is an open map. 
    \end{enumerate}
  \end{proof}
The next lemma collects the crucial properties of the final topology of the spaces $X_L$. In addition to that, the Borel $\sigma$-algebras of these spaces are determined. 
\begin{lemma} 
  \label{lemma:reltop}
  \begin{enumerate}
  \item
    \label{lemma:reltop1}
    The subspace topologies of $\im[f]$ and $\TO^{|L|}$ w.r.t.\ the final topology $\T_F$ on $X_L$ are given by 
    $\T_f$ and $\T_L$, respectively. 
    For a subset $U\subseteq \im[f]$ we have $U\in \T_f$ iff $U$ is open in $X_L$.
  \item
    \label{lemma:reltop2}
    $X_L$ is a compact Hausdorff space.
  \item
    \label{lemma:reltop3}
    We have $\Borel\left(X_L\right)=\Borel(\im[f])\sqcup\BTK$ and 
    \begingroup
    \setlength{\leftmarginii}{20pt}
    \begin{enumerate}
    \item[(a)]
      If $\mu$ is a finite Radon measure on $\Borel(X_L)$, then $\mu|_{\Borel(\im[f])}$ and $\mu|_{\Borel(\TO^{|L|})}$ are finite Radon measures as well. 
    \item[(b)]
      If $\mu_{f}\colon \Borel(\im[f]) \rightarrow [0,\infty)$ and $\mu_{T}\colon \BTK \rightarrow [0,\infty)$ are finite Radon measures, then 
      \begin{align*}
 		(\mu_f\oplus \mu_{\TO})(A):=\mu_{f}(A\cap \im[f])+ \mu_{\TO}\big(\raisebox{-1pt}{$A\cap \TO^{|L|}$}\big)
        \qquad\forall\:  A\in \Borel(X_L)	
      \end{align*}
      is a finite Radon measure on $\Borel\big(X_L\big)$.
    \end{enumerate}
    \endgroup
  \end{enumerate}
  \begin{proof}
    \begin{enumerate}
    \item
      We first collect the following facts: 
      \begin{enumerate}
      \item[(a)]
 	The topology on $\qR$ induces the standard topologies on $\RR$ and $\RB$.
      \item[(b)]
        $U\in \T_F$ iff $\pi_L^{-1}(U)$ is open in $\qR$.
      \item[(c)]
        $W\subseteq \RR$ is open in $\qR$ iff $W$ is open in $\RR$.
      \item[(d)]
     	If $B\subseteq \RB$ is open, then there is $U\subseteq \im[f]$ such that $f^{-1}(U)\sqcup B$ is open in $\qR$.\footnote{By (e), this is equivalent to show that we find an open subset $W\subseteq \RR$ such that $W\sqcup B$ is open in $\qR$. But, this is clear if $B=\chih_l^{-1}(\OO)$ for some open subset $\OO\subseteq \TO$ and $l\in \RR$ (see \textbf{Type 3} sets defined in Subsection \ref{subsec:Rquer}). Now, since the sets of the form  $\chih_l^{-1}(\OO)$ provide a subbasis for the topology on $\RB$, the claim follows.}
      \item[(e)]
	$f\colon \RR\rightarrow \im[f]$ is a homeomorphism.
      \item[(f)]
        $\pih_L\colon \RB\rightarrow \TO^{|L|}$ is continuous and open.
      \end{enumerate}
	\vspace{4pt}      
      
      \noindent
      We show the statements concerning the subspace topologies: 
      
	\vspace{3pt} 
      
      $\boldsymbol{\im[f]}$:\hspace{4pt}Let $U\subseteq \im[f]$. Then:
      \vspace{2pt}
      
      \qquad\quad\:\hspace{10pt}
      $U$ is open w.r.t.\ the topology inherited from $X_L$

      \qquad$\Longleftrightarrow$\:\:\:$\exists\: V\subseteq \TO^{|L|}$ such that $U\sqcup V$ is open in $X_L$

      \qquad$\Longleftrightarrow$\:\:\:$\exists\: V\subseteq \TO^{|L|}$ such that $\pi_L^{-1}(U\sqcup V)$ is open in $\qR$\hspace*{\fill}{(b)}

      \qquad$\Longleftrightarrow$\:\:\:$\exists\: V\subseteq \TO^{|L|}$ such that $f^{-1}(U)\sqcup \pih_L^{-1}(V)$ is open in $\qR$
      
      \qquad$\Longleftrightarrow$\:\:\:$f^{-1}(U)$ is open in $\RR$\hspace*{\fill}{(c)}

      \qquad$\Longleftrightarrow$\:\:\:$U\in \T_f$  \hspace*{\fill}(e)
      
      \vspace{2ex}
      $\boldsymbol{\TO^{|L|}}$:\hspace{4pt}Let $V\subseteq \TO^{|L|}$. Then:
      \vspace{2pt}
      
      \qquad\quad\hspace{10pt}
      $V$ is open w.r.t.\ the topology inherited from $X_L$

      \qquad$\Longleftrightarrow$\:\:\:$\exists\: U\subseteq \im[f]$ such that $U\sqcup V$ is open in $X_L$

      \qquad$\Longleftrightarrow$\:\:\:$\exists\: U\subseteq \im[f]$ such that $\pi_L^{-1}(U\sqcup V)$ is open in $\qR$\hspace*{\fill}{(b)}

      \qquad$\Longleftrightarrow$\:\:\:$\exists\: U\subseteq \im[f]$ such that $f^{-1}(U)\sqcup \pih_L^{-1}(V)$ is open in $\qR$

      \qquad$\Longleftrightarrow$\:\:\:$\pih_L^{-1}(V)$ is open in $\RB$\hspace*{\fill}{(a),(d)}

      \qquad$\Longleftrightarrow$\:\:\:$V\in \T_L$  \hspace*{\fill}(f) 			
      
      \vspace{1ex}
      Finally, observe that $\im[f]$ is open in $X_L$ because $\pi_L^{-1}(\im[f])=\RR$ is open in $\qR$. Thus, $U\subseteq \im[f]$ is open in $X_L$ iff $U$ is open w.r.t.\ the subspace topology on $\im[f]$ inherited from $X_L$. Since this topology equals $\T_f$, the claim follows.
    \item
      The spaces $X_L$ are compact by compactness of $\qR$ and continuity of $\pi_L$. For the Hausdorff property observe that $\T_F$ contains the sets 
      \begin{align*}
        \begin{array}{lcrclcl}
          \textbf{Type 1':} && f(V) & \!\!\!\sqcup\!\!\! & \emptyset 
          && \text{with open $V \subseteq \RR$,} \\
          \textbf{Type 2':} && f(K^c) & \!\!\!\sqcup\!\!\! & \TO^{|L|}
          && \text{with compact $K \subseteq \RR$,} \\
          \textbf{Type 3':} && f\big(\chi_{l_i}^{-1}(\OO)\big) & \!\!\!\sqcup\!\!\! & \pr_i^{-1}(\OO)
          && \text{with $\OO\subseteq \TO$ open, $1\leq i\leq  k$.}
        \end{array}
      \end{align*}
      for $\pr_i\colon \TO^k\rightarrow \TO$,\:\: $(s_1,{\dots},s_k)\mapsto s_i$ the canonical projection map. In fact, as one easily verifies, the preimage of a set of \textbf{Type m'} is a subset of $\qR$ of \textbf{Type m}, cf.\ Subsection \ref{subsec:Rquer}.
       
      Then, by injectivity of $f$, the elements of $\im[f]$ are separated by sets of \textbf{Type 1'}. Moreover, if $x\in \im[f]$ and $(s_1,{\dots},s_k)\in \TO^k$, we can choose a relatively compact neighbourhood $W$ of $f^{-1}(x)$ in $\RR$, and define $U:= f(W)$ as well as $V:=f\raisebox{1pt}{$\big($}\raisebox{-1pt}{${\ovl{W}\hspace{0.9pt}}^c$}\raisebox{1pt}{$\big)$}\sqcup \TO^{|L|}$. Finally, if $(s_1,{\dots},s_k),(s'_1,{\dots},s'_k)\in \TO^k$ are different elements, then $s_i\neq s'_i$ for some $1\leq i\leq k$. Then, for open neighbourhoods $\OO,\OO' \subseteq \TO$ of $s_i$ and $s_i'$, respectively, with $\OO \cap \OO' =\emptyset$, we have 	
	\begin{align*}      
      \big[f\big(\chi_{l_i}^{-1}(\OO)\big) \sqcup \pr_i^{-1}(\OO)\big]\cap \big[f\big(\chi_{l_i}^{-1}(\OO')\big) \sqcup \pr_i^{-1}(\OO')\big]=\emptyset
    \end{align*}
    by injectivity of $f$.
    \item  
      We repeat the arguments from the proof of Lemma \ref{lemma:Radon}:
      
      If $U\subseteq X_L$ is open, then $U\cap \im[f]\in \T_f$ and $U\cap \TO^{|L|}\in \T_L$ 
      by Part \ref{lemma:reltop1}). This shows $U\in \Borel(\im[f])\sqcup\BTK$, i.e., $\Borel\left(X_L\right)\subseteq\Borel(\im[f])\sqcup\BTK$ as the right hand side is a $\sigma$-algebra. For the converse inclusion, recall that $U\in \T_f$ iff $U\subseteq \im[f]$ is open in $X_L$ again by Part \ref{lemma:reltop1}). Thus, $\Borel(\im[f])\subseteq \Borel(X_L)$. Finally, if $A\subseteq \TO^{|L|}$ is closed, then $A$ is compact w.r.t.\ $\T_L$. This means that $A$ is compact w.r.t.\ the subspace topology inherited from $X_L$, implying that $A$ is compact as a subset of $X_L$. Then $A$ is closed by the Hausdorff property of $X_L$, so that $A\in \Borel\left(X_L\right)$, hence $\BTK\subseteq \Borel\left(X_L\right)$. 
      \begingroup
      \setlength{\leftmarginii}{18pt}
      \begin{enumerate}
      \item[(a)]
        The measures $\mu|_{\Borel(\im[f])}$ and $\mu|_{\Borel(\TO^{|L|})}$ are well defined and finite. Their inner regularities follow from the fact that subsets of $\im[f]$ and $\TO^{|L|}$ are compact w.r.t.\ $\T_f$ and $\T_L$, respectively, iff they are so w.r.t.\ the topology on $X_L$, just by Part \ref{lemma:reltop1}).
      \item[(b)]   
        Obviously, $\mu_f\oplus \mu_{\TO}$ is a finite Borel measure, and its inner regularity follows by a simple $\epsilon\slash 2$-argument from inner regularities of $\mu_{f}$ and $\mu_{\TO}$.
      \end{enumerate}
      \endgroup
    \end{enumerate}
  \end{proof}
\end{lemma}
Combining the Lemmata \ref{prop:OpenMapp} and \ref{lemma:reltop}, we obtain
\begin{proposition}
  \label{th:projlim}
  \begin{enumerate}
  \item
    $\qR$ is a projective limit of $\{X_L\}_{L\in I}$.
  \item
    A family $\{\mu_L\}_{L \in I}$ of measures $\mu_L$ on $X_L$ is a consistent family of normalized Radon measures w.r.t.\ $\{X_L\}_{L\in I}$ iff the following holds:
    \begingroup
    \setlength{\leftmarginii}{20pt}
    \begin{enumerate}
    \item
      There is $t\in [0,1]$, such that for each $L\in I$ 
      we have
      \begin{align*}	
        \mu_L=t\:\mu_f\hs\oplus\hs (1-t)\:\mu_{\TO,L}
      \end{align*} 
      for $\mu_f$ and $\mu_{\TO,L}$ normalized\footnote{If $t$ equals $0$ or $1$, we allow $\mu_f=0$ or $\mu_{\TO,L}=0$, respectively.} Radon measure on $\im[f]$ and $\TO^{|L|}$, respectively.
    \item
      For all $L,L'\in I$ with $L\leqZ L'$, we have $\pih^{L'}_{L}(\mu_{\TO,L'})=\mu_{\TO,L}$.
    \end{enumerate}
    \endgroup
  \end{enumerate}
\end{proposition}
\begin{proof}
    \begin{enumerate}
    \item
      The spaces $X_L$ are compact and Hausdorff by Lemma \ref{lemma:reltop}.\ref{lemma:reltop2}. 
      Moreover, each  $\pi_L$ is surjective by Lemma \ref{prop:OpenMapp} and surjectivity of $f$.\ref{lemma:OpenMapp4}. 
      If $L,L'\in I$ with $L\leqZ L'$, continuity of $\pi^{L'}_{L}$ is clear from
      $\pi^{L'}_{L}\cp \pi_{L'}=\pi_{L}$, which, in turn, is immediate from  multiplicativity of the functions $\chi_l$. Finally, Condition \ref{def:ProjLim3}) from
      Definition \ref{def:ProjLim} follows from injectivity of $f$ and the fact that the functions $\{\chi_l\}_{l\in \mathbb{R}}$ generate $\CAP$.
    \item
 Let $\{\mu_L\}_{L\in I}$ be a consistent family of normalized Radon measures w.r.t.\ $\{X_L\}_{L\in I}$. Then, Lemma \ref{lemma:reltop}.\ref{lemma:reltop3} shows that for each $L\in I$ we have
      \begin{align*}
 		\mu_L=    \mmu_{f,L}\oplus \mmu_{\TO,L}
      \end{align*}
      for $\mmu_{f,L}$ and $\mmu_{\TO,L}$ finite Radon measures on $\im[f]$ and $\TO^{|L|}$, respectively. Moreover, consistency enforces $\mmu_{f,L}=\mmu_{f,L'}$ for all $L,L'\in I$. In fact, by Lemma \ref{lemma:normRM} there is a unique normalized Radon measure $\mu$ on $\qR$ with $\mu_L=\pi_L(\mu)$ for all $L\in I$. Consequently, for each $A\in \Borel(\im[f])$ and all $L\in I$ we have	  
      \begin{align*}     	
        \mmu_{f,L}(A)=\mu_{L}(A)=\pi_L(\mu)(A)= \mu\big(f^{-1}(A)\big)=:\mmu_f(A).
      \end{align*}
      By the same arguments, (b) follows from the consistencies of the measures $\{\mu_L\}_{L\in I}$. 
      Finally, if $t:= \mmu_f(\im[f])\in (0,1)$, then (a) holds for $\mu_f:=\frac{1}{t}\:\mmu_f$ and $\mu_{\TO,L}:=\frac{1}{1-t}\:\mmu_{\TO,L}$ for $L\in I$.
      If $t=0$, we let $\mu_{\TO,L}:=\mmu_{\TO,L}$ for all $L\in I$ and if $t=1$, we define $\mu_f:=\mmu_f$.

      For the converse implication, let
      $\{\mu_L\}_{L\in I}$ be a family of measures $\mu_L$ on $X_L$ such that (a) and (b) hold. Then, Lemma \ref{lemma:reltop}.\ref{lemma:reltop3} shows that each $\mu_L$ is a finite Radon measure, and obviously we have $\mu_L(X_L)=1$. 
      Finally, from (b), for $A\in \Borel(X_L)$ we obtain 
      \begin{align*}
        \pi^{L'}_{L}(\mu_{L'})(A)&=\mu_{L'}\left(\pillstr^{-1}(A)\right)\\[-4pt]
        &= t\: \mu_f\left(\pillstr^{-1}(A)\cap \im[f]\right)+(1-t)\:\mu_{\TO,L'} \left(\pillstr^{-1}(A)\cap \TO^{|L'|}\right)\\[-4pt]
        &=t\: \mu_f\left(A\cap \im[f]\right)+ (1-t)\:\mu_{\TO,L'} \Big(\!\left(\raisebox{-0.1ex}{$\pih^{L'}_{L}$}\right)^{-1}\left(\raisebox{-0.1ex}{$A\cap \TO^{|L|}$}\right)\!\Big)\\[-4pt]
        &=t\: \mu_f\left(A\cap \im[f]\right)+(1-t)\:\pih^{L'}_{L}(\mu_{\TO,L'}) \big(\raisebox{-1pt}{$A\cap \TO^{|L|}$}\big)\\[-1pt]
        &=t\: \mu_f\left(A\cap \im[f]\right)+(1-t)\:\mu_{\TO,L} \big(\raisebox{-1pt}{$A\cap \TO^{|L|}$}\big)\\[-1pt]
      	&=\mu_{L}(A).
      \end{align*}   
    \end{enumerate}		
\end{proof}

\subsection{Radon Measures on $\boldsymbol{\qR}$}
\label{subsec:CylMeas}
In this final subsection, we will use the results of the previous part in order to fix normalized Radon measures on $\qR$. Due to
Proposition \ref{th:projlim}, this can be done as follows:
\begingroup
\setlength{\leftmargini}{13pt}
\begin{enumerate} 
\item[{\bf 1}]
  Determine a family of normalized Radon measures $\{\mu_{\TO,L}\}_{L\in I}$ on $\TO^{|L|}$ that fulfil condition (b).
\item[{\bf 2}]
  Fix an injective and nowhere vanishing element $f\in C_0(\RR)$ with suitable image, together with a normalized Radon measure $\mu_f$ on $\im[f]$. 
\item[{\bf 3}]
  Adjust $t\in [0,1]$. 
\end{enumerate}
\endgroup
\vspace{3pt}

\noindent
{\bf Step 1}
\newline
\vspace{-2.3ex}
\newline
We choose 
$\mu_{\TO,L}$ to be the Haar measure $\mu_{|L|}$ on $\TO^{|L|}$ because
\begingroup
\setlength{\leftmargini}{18pt}
\begin{itemize}
\item
  \itspace
  This is canonical from the mathematical point of view and in analogy to the case $\RB$ \cite{Vel}, where this choice
  results in the Haar measure on this space, cf.\ Subsection \ref{subsec:Motivation}.
\item
  \itspace
  These measures will suggest a natural choice of $f$ and $\mu_f$ in \textbf{Step 2}.
\end{itemize}
\endgroup
\begin{lemma}
  \label{lemma:Projm}
  Let $\mu_f\colon \Borel(\im[f])\rightarrow [0,1]$ be a normalized Radon measure and $t\in[0,1]$. For each $L\in I$ let 
  \begin{align*}
 	 \mu_L:= t\: \mu_f\hs\oplus\hs (1-t)\: \mu_{|L|}. 
  \end{align*} 
  Then, $\{\mu_L\}_{L\in I}$ is a consistent family of normalized Radon measures, and the corresponding normalized Radon measure on $\qR$ is given by 
  \begin{align} 
    \label{eq:Radmeas}
    \mu:= t\: f^{-1}(\mu_f)\hs \oplus\hs (1-t)\:\muB.
  \end{align} 
  \end{lemma}
  \begin{proof}
      Let $L\in I$, $A \in \Borel(X_L)$ and $\mu$ be defined by \eqref{eq:Radmeas}. Then,
    \begin{align*}
      \pi_L(\mu)(A)&=t\: f^{-1}(\mu_f)\big(f^{-1}(A\cap \im[f])\big) + (1-t)\: \mu_{\mathrm{Bohr}}\big(\pih_L^{-1}\big(\raisebox{-1pt}{$A\cap \TO^{|L|}$}\big)\big)\\
      &= t\:\mu_f(A\cap \im[f])+ (1-t)\:\pih_L(\mu_{\mathrm{Bohr}})\big(\raisebox{-1pt}{$A\cap \TO^{|L|}$}\big).
    \end{align*}
    Thus, if we know that 
    $\pih_L(\mu_{\mathrm{Bohr}})=\mu_{|L|}$ holds for all $L\in I$, the claim follows. In fact, then consistency of $\{\mu_L\}_{L\in I}$ is automatically fulfilled because $\mu$ is a well-defined normalized Radon measure on $\qR$.
    
    Now, in order to show $\pih_L(\mu_{\mathrm{Bohr}})=\mu_{|L|}$, it suffices to show the translation invariance of the normalized Radon measure $\pih_L(\mu_{\mathrm{Bohr}})$.
    For this, let $\tau \in \TO^{|L|}$. Then, by surjectivity of $\pih_L$ we find $\psi \in \RB$ with $\pih_L(\psi)=\tau$. Since $\pih_L$ is a homomorphism w.r.t.\ the group structure on $\RB$, for $A\subseteq \TO^{|L|}$ we have
    \begin{align*}
      \pih_L\left(\psi+ \pih_L^{-1}(A)\right)=\pih_L(\psi)\cdot \pih_L\left(\pih_L^{-1}(A)\right)=\tau \cdot A.
    \end{align*}
    Then, applying $\pih_L^{-1}$ to both sides gives $\psi+ \pih_L^{-1}(A)\subseteq \pih_L^{-1}(\tau \cdot A)$. For the opposite inclusion, let $\psi'\in \pih_L^{-1}(\tau\cdot A)$. Then, $\psi'- \psi \in \pih_L^{-1}(A)$ because 
    $\pih_L(\psi'- \psi)=\tau^{-1}\cdot \pih_L(\psi')\in A$. Consequently, $\psi'\in \psi+ \pih_L^{-1}(A)$, hence $\pih_L^{-1}(\tau\cdot A)\subseteq \psi+ \pih_L^{-1}(A)$. Thus, $\psi+ \pih_L^{-1}(A)=\pih_L^{-1}(\tau\cdot A)$ so that
    \begin{align*}
      \pih_L(\mu_{\mathrm{Bohr}})(\tau\cdot A)&=\mu_{\mathrm{Bohr}}\left(\pih_L^{-1}(\tau\cdot A)\right)=\mu_{\mathrm{Bohr}}\left(\psi +\pih_L^{-1}(A)\right)\\
      &= \mu_{\mathrm{Bohr}}\left(\pih_L^{-1}(A)\right)=\pih_L(\mu_{\mathrm{Bohr}})(A)
    \end{align*}
    for all $A\in \BTK$ shows that $\pih_L(\mu_{\mathrm{Bohr}})$ is translation invariant. 
  \end{proof}
\vspace{-2pt}
\noindent
{\bf Step 2}
\newline
\vspace{-2.3ex}
\newline
If $f,f'\in C_0(\mathbb{R})$ both are injective and vanish nowhere, the respective projective structures from Definition \ref{def:ProjLimit} are equivalent in the sense that the corresponding spaces $X_L, X'_L$ are homeomorphic via the structure preserving maps $\Omega_L\colon X_L\rightarrow X'_L$ defined by
\begin{align*}
	 \qquad \Omega_L|_{\TO^{|L|}}:=\id_{\TO^{|L|}}\quad\qquad\text{and}\quad\qquad \Omega_L|_{\im[f]}:=f'\cp f^{-1}.
\end{align*}
 Moreover, if $\mu_f$ is a normalized Radon measure on  $\im[f]$, then $\mu_{f'}:=\big(f'\cp f^{-1}\big)(\mu_f)$ is a normalized Radon measure on $\im[f']$, and it is clear from \eqref{eq:Radmeas} that the corresponding Radon measures $\mu,\mu'$ on $\qR$
from Lemma \ref{lemma:Projm} coincide. 

All this makes sense because 
in contrast to $\CAP$, where we have the canonical generators $\{\chi_l\}_{l\in \mathbb{R}}$, there is no distinguished (nowhere vanishing, injective) generator $f$ in $C_0(\mathbb{R})$. But, this also means that, in order to fix some measure on $\qR$, we can restrict to functions $f$ with a reasonable image; such as the ``shifted'' unitcircle $\Ts:=1+ \TO\backslash\{-1\}\subseteq \mathbb{C}$. In fact, here the analogy to $\TO^{|L|}$ suggests to take the Haar measure $\mu_1$ on $\TO$. Thus, in the sequel, we will restrict to 
\begin{align*}
	\F:=\{f\in C_0(\RR)\:|\: \im[f]=\Ts\} 
\end{align*}	
	 whereby for each $f\in \F$ we define $\mu_f:=\mus\colon \Borel\big(\Ts\big)\rightarrow [0,1]$. Here, 
\begin{align*}	
	 \mus(A):=\mu_1(A-1)=+_1(\mu_1)\qquad\forall\: A\in \Borel\big(\Ts\big)
\end{align*}	 
with $+_1\colon \TO\backslash\{-1\}\ni z\mapsto z+1\in \Ts$.
	  It follows that  
\begin{align}
  \label{eq:LebesgueHomeo}
  \big\{f^{-1}(\mu_f)\:\big|\: f\in \F\big\}=\big\{\rho(\muL)\:\big|\: \rho\in \HHH\big\}
\end{align}	
for $\HHH$ the set of homeomorphisms $\rho\colon (0,1)\rightarrow \RR$, $\lambda$ the Lebesgue measure on $\RR$, and $\rho(\lambda)$ the push forward of $\lambda|_{\Borel((0,1))}$ by $\rho$.
\newline
\vspace{-8pt}
\newline
{\bf Proof of \eqref{eq:LebesgueHomeo}:} 
We consider the function $h\colon(0,1]\ni t\mapsto e^{\I\hspace{1pt} 2\pi[t-1\slash 2]} \in \TO$. Then,
$\mu_1=h(\lambda)|_{\Borel(\TO)}$, and for $f\in \F$ we have $f^{-1}(\mu_f)= \rho(\lambda)$ for $\HHH\ni \rho:=f^{-1}\cp+_1\cp h|_{(0,1)}$. Conversely, if $\rho \in \HHH$, then $\rho(\lambda)=f^{-1}(\mus)$ for $\F\ni f:=+_1\cp h\cp \rho^{-1}$.\hspace*{\fill}{\scriptsize$\blacksquare$}
\newline
\vspace{-4pt}
\newline
So, if we restrict to projective structures arising from elements $f\in \F$, Lemma \ref{lemma:Projm} and \eqref{eq:LebesgueHomeo} select the normalized Radon measures 
\begin{align*}
  \mu_{\rho,t}= t\:\rho(\muL)\:\oplus\: (1-t)\:\mu_{\mathrm{Bohr}}
\end{align*}
for $\rho\colon (0,1)\rightarrow \mathbb{R}$ a homeomorphism and $t\in [0,1]$. 
\vspace{10pt}
  
\noindent
{\bf Step 3}
\vspace{6pt}

\noindent
To adjust the parameter $t\in[0,1]$, we take a look at the $L^2$-Hilbert space  $\Hil_{\rho,t}$ that correspond to $\mu_{\rho,t}$. 
\begin{lemma}
  \label{lemma:techlemma}
  For $A\in \BRq$, let $\chi_A$ denote the corresponding characteristic function.
  \begin{enumerate}
  \item
    \label{lemma:techlemma1}
    If $\rho_1,\rho_2\colon (0,1)\rightarrow \RR$ are homeomorphisms and $t_1,t_2 \in (0,1)$, then 
    \begin{align*}
      \begin{split}
        \varphi \colon \Lzw{\qR}{\mu_{\rho_1,t_1}}&\rightarrow \Lzw{\qR}{\mu_{\rho_2, t_2}}\\
        \psi & \mapsto  \sqrt{\frac{t_1}{t_2}} \:(\chi_{\RR}\cdot \psi)\cp \big(\rho_1\cp \rho_2^{-1}\big) + \sqrt{\frac{(1-t_1)}{(1-t_2)}}\:\chi_{\RB}\cdot \psi
      \end{split}
    \end{align*}
    is an isometric isomorphism. The same is true for 
    \begin{align*}  
      &\varphi \colon \Hil_{\rho_1,1}\rightarrow \Hil_{\rho_2,1},\quad \psi \mapsto (\chi_{\RR}\cdot \psi)\cp \big(\rho_1\cp \rho_2^{-1}\big),\\ 
      &\varphi \colon \Hil_{\rho_1,0}\rightarrow \Hil_{\rho_2,0},\quad \psi \mapsto  \psi. 
    \end{align*}
  \item
    \label{lemma:techlemma2}
    If $t=1$, then $\Hil_{\rho,1}\cong \Lzw{\RR}{\rho(\muL)}\cong \Lzw{\RR}{\muL}$ for each $\rho\in H$, whereby $\cong$ means canonically isometrically isomorphic.
  \end{enumerate}
  \end{lemma}
  \begin{proof}
    \begin{enumerate}
    \item
      This  is immediate from the general transformation formula.
    \item
      The first isomorphism is because $\mu_{\rho,0}(\RB)=0$. Then, by the first part, it suffices to specify the second isomorphism for the case that $\rho$ is a diffeomorphism. But, in this case we have $\rho(\lambda)=\frac{1}{|\dot\rho|}\hs\lambda$, so that for 
	\begin{align*}
		\textstyle\varphi\colon \Lzw{\RR}{\rho(\muL)}\rightarrow \Lzw{\RR}{\muL},\quad \psi \mapsto  \frac{1}{\sqrt{|\dot \rho|}}\hs\psi
	\end{align*}  we have 
      \begin{align*}
        \langle \varphi(\psi_1),\varphi(\psi_2)\rangle_\lambda&=\int_{\RR}  \psi_1 \ovl{\psi_2} \:\frac{1}{|\dot \rho|}\: \dd\muL= \int_{\RR}  \psi_1 \ovl{\psi_2} \:\: \dd\rho(\muL)
        =\langle \psi_1,\psi_2\rangle_{\rho(\lambda)}.
      \end{align*}
    \end{enumerate} 
  \end{proof}
So, if $\rho_0\in H$ and $t_0\in (0,1)$ are fixed, Lemma \ref{lemma:techlemma} shows that up to \emph{canonical} isometrical isomorphisms the parameters $\rho$ and $t$ give rise to the following three Hilbert spaces:
\begingroup
  \setlength{\leftmargini}{20pt} 
  \begin{enumerate}
  \item[1)]
    $\Hil_{\rho,1}\cong \Lzw{\RR}{\muL} \cong \Lzw{\RR}{\rho(\muL)}$\:\quad for all\quad $\rho\in \HHH$,\hspace*{\fill}{(Lemma \ref{lemma:techlemma}.\ref{lemma:techlemma2})}
  \item[2)] 
    $\Hil_{\rho,t}\cong L^2\big(\hspace{1pt}\raisebox{-0.1ex}{$\qR$},\mu_{\rho_0,t_0}\big)$\hspace{49pt}\quad for all\quad $\rho\in \HHH$, $t\in (0,1)$,\hspace*{\fill}{(Lemma \ref{lemma:techlemma}.\ref{lemma:techlemma1})}
  \item[3)]
    $\Hil_{\rho,0}\cong \Lzw{\RB}{\mu_{\mathrm{Bohr}}}$\hspace{31.5pt}\quad for all\quad $\rho\in \HHH$. \hspace*{\fill}{($\RR$ is of measure zero)}
  \end{enumerate}
  \endgroup
  \noindent
  Here, the Hilbert spaces in $2)$ and $3)$ are isometrically isomorphic because their Hilbert space dimensions coincide. In contrast to that, 
  the spaces in $1)$ and $3)$ cannot be isometrically isomorphic because $\Lzw{\RR}{\muL}$ is separable and $\Lzw{\RB}{\mu_{\mathrm{Bohr}}}$ is not so. 
  
\section{Conclusions}
\begingroup
\setlength{\leftmargini}{14pt}
\begin{itemize}
\item
In the first part of this paper, we have reformulated the definition of a projective limit of a compact Hausdorff space in order to make it more convenient for defining normalized Radon measures thereon. This, we have applied to the cosmological quantum configuration space $\RR\sqcup\RB$ for which, using Haar measures on tori, we have motivated the measures 
  $\mu_{\rho,t}=t\:\rho(\muL)\hs \oplus\hs  (1-t)\:\mu_{\mathrm{Bohr}}$.
Up to canonical isometric isomorphisms, these give rise to the Hilbert spaces 
\begin{align*}
	 \Lzw{\RR}{\muL} \qquad\qquad\qquad L^2\big(\hspace{1pt}\raisebox{-0.1ex}{$\qR$},\mu_{\rho,t}\big) \qquad\qquad\qquad \Lzw{\RB}{\mu_{\mathrm{Bohr}}}
\end{align*}
for $t\in (0,1)$ and $\rho\colon (0,1)\rightarrow \RR$ fixed. Although the last two Hilbert spaces are isometrically isomorphic, there might exist interesting representations of the reduced holonomy-flux algebra on $L^2\big(\hspace{1pt}\raisebox{-0.1ex}{$\qR$},\mu_{\rho,t}\big)$ which are not unitarily equivalent to the standard representation \cite{MathStrucLQG} on $\Lzw{\RB}{\mu_{\mathrm{Bohr}}}$.
\item
 In Subsection \ref{subsec:Rquer}, we have shown that the only Radon measure on $\qR$ which fulfils the two conditions in \cite{Jon} is given by $0_\RR\oplus \muB$. So, it should be investigated, whether the introduced projective structures can be used to single out this measure by certain invariance properties as well. This would provide an analogue to the uniqueness result for homogeneous LQC proven in \cite{ASHPARAM}. 
\end{itemize}
\endgroup 

\section*{Acknowledgements}
The author thanks Benjamin Bahr and Christian Fleischhack for numerous discussions
and many helpful comments on a draft of the present paper. He is grateful to two anonymous JMAA referees whose competent suggestions helped him to enhanced the presentation of the paper. 
 The author has been supported by the Emmy-Noether-Programm of the Deutsche Forschungsgemeinschaft under grant FL~622/1-1.


\appendix

\section*{Appendix}
\section{Proof of the Lemma in Subsection \ref{subsec:Motivation}}
\label{sec:ProofOfLemmaImCirc}
\noindent
{\bf Proof of Lemma \ref{lemma:BildCirc}:}
Let $\sigma\in \SU$ with $\varrho(\sigma)(\vec{e}_3)=\vec{n}$.
\begin{enumerate}
\item[\ref{lemma:BildCirc2})]	
  It suffices to show the claim for $\RR\subseteq \qR$. Now, if $c\in \mathbb{R}\subseteq \qR$, then 
  \begin{align}  
  \label{eq:equ111}    
    (\oishomc\cp \xi\hspace{1pt})(c)\stackrel{\eqref{eq:iotaxiR}}{=}(\oishomc\cp \iota_\RR)(c)\stackrel{\eqref{eq:commdiag}}{=}(\iota_{\Con} \cp \ishomc)(c), 
  \end{align}
  hence 
    \begin{align}
    \label{eq:imR}
    \begin{split}
      \pi_\delta(c)
      & \stackrel{\eqref{eq:projdef}}{=} \Delta(c)(\delta)=(\homiso\cp\oishomc\cp\xi)(c)(\delta) \stackrel{\eqref{eq:equ111}}{=}(\homiso\cp\iota_{\Con} \cp \ishomc)(c)(\delta)      \\
      &\hspace{2pt}\stackrel{\eqref{eq:SpeczuHomm}}{=} h_\delta\left((\iota_{\Con} \cp \ishomc)(c)\right) 
      \stackrel{\eqref{eq:patralla}}{=}\left(\pr_2\cp\parall{\delta}{\ishomc(c)}\right)(x,\me)
       \stackrel{\eqref{eq:parallcirc}}{=} \exp\left(\textstyle\frac{\tau}{2}\cdot \vec{n}\right)\cdot \Co{\sigma}(A(\tau,c)).
    \end{split}
  \end{align}
	Thus, $\pi_\delta|_{\RR}$ is continuous and $\pi_\delta(\mathbb{R})\subseteq \SU$ is connected. 
	
	Now, each proper Lie subalgebra of $\mathfrak{su}(2)$ is of dimension $1$ and, since $\SU$ is connected, the same is true for the proper Lie subgroups of $\SU$. Let $H\subseteq\SU$ be such a proper Lie subgroup with $\pi_\delta(\RR)\subseteq H$. Then, $\mathfrak{h}=\spann_{\mathbb{R}}(\vec{s}\hspace{2pt})$ for some $\vec{s}\in\mathfrak{su}(2)$, and $\TO_{\vec{s}}$ is the unique connected Lie subgroup of $H$ with Lie algebra $\mathfrak{h}$, i.e., the component of $\me$ in $H$. 
  But $\me=\pi_\delta(0_{\mathbb{R}})\in \pi_\delta(\mathbb{R})\cap H$, hence $\pi_\delta(\mathbb{R})\subseteq \TO_{\vec{s}}$ by connectedness of $\pi_\delta(\mathbb{R})$. Thus, each two elements of $\pi_\delta(\mathbb{R})$ have to commute, which is impossible by the second line in \eqref{eq:matrixentr}. 
\item[\ref{lemma:BildCirc4})]
  If $\ovl{x}\in \RB\subseteq \qR$, then 
  \begin{align*}
    \begin{split} 
      \pi_\delta(\ovl{x})&=\homiso\left(\left(\oishomc\cp \xi\right)(\x)\right)(\delta)
      \stackrel{\eqref{eq:pigammawquer}}{=} \left(\left(\oishomc\cp \xi\right)(\x)\big([h_\gamma(\cdot)]_{ij}\big)\right)_{ij} \\
      &\hspace{0pt}=d\cdot \Co{\sigma}\left( 
      \begin{pmatrix} \xi(\x)\left(c\mapsto \cos(\pc \tau)+\frac{\I}{2\pc}\sin(\pc \tau)\right) &\xi(\x)\left(c\mapsto\frac{cr}{\pc} \sin(\pc \tau)\right) \\ \xi(\x)\left(c\mapsto-\frac{cr}{\pc}\sin(\pc \tau)\right) & \xi(\x)\left(c\mapsto\cos(\pc \tau)-\frac{\I}{2\pc}\sin(\pc \tau)\right)  \end{pmatrix}
      \right)\\
      &\hspace{-0pt}=   d\cdot\Co{\sigma} \left(
      \begin{pmatrix} \xi(\x)\left(c\mapsto \cos(c r\tau)\right) &\xi(\x)\left(c\mapsto\sin(c r \tau)\right) \\ \xi(\x)\left(c\mapsto-\sin(c r \tau)\right) & \xi(\x)\left(c\mapsto\cos(cr \tau)\right)  \end{pmatrix}
      \right)   \\          
      &\hspace{-0pt}=d 
      \cdot\Co{\sigma}\left(\big(\ovl{x}\big(c\mapsto \exp(-c r\tau \cdot \tau_2)_{ij}\big) \big)_{ij}\right)\in d\cdot\Co{\sigma}(\TO_{\vec{e}_2})=d \cdot \TO_{\dot{\delta}(0)}.
    \end{split}
  \end{align*}
  Here, the third step follows from\eqref{eq:pigammawquer}, \eqref{eq:parallcirc}, \eqref{eq:matrixentr}, and multiplicativity and linearity of $\xi(\x)$. The fourth step is because the (unique) decompositions of the matrix entries  
  \begin{align*}
    a\colon c\mapsto \cos(\pc \mm)+\textstyle\frac{i}{2\pc}\sin(\pc \mm)\qquad\quad \text{and}\qquad\quad
    b\colon c\mapsto \textstyle\frac{cr}{\pc} \sin(\pc \mm)
  \end{align*}
  in \eqref{eq:matrixentr} into direct sums of the form $C_{0}(\mathbb{R})\oplus \CAP$ are given by
  \begin{align}
    \label{eq:cappluco}
    a(c)&= \left(\cos(\pc \mm)+\textstyle\frac{i}{2\pc}\sin(\pc \mm)- \cos(cr\tau)\right) + \cos(cr\tau)\\
    \label{eq:cappluco2}
    b(c)&=\left(\textstyle\frac{cr}{\pc} \sin(\pc \mm)-\sin(cr\tau)\right) + \sin(cr\tau).
  \end{align} 
  Finally, the last step is due to \eqref{eq:expformel}; and that $\pi_\delta(\ovl{x})$ is contained in $d \cdot \TO_{\dot{\delta}(0)}$ follows from closedness of $\TO_{\vec{e}_2}$ and that we find a net $\{c_\alpha\}_{\alpha\in I}\subseteq \RR$ such that for each $1\leq i,j\leq 2$ we have 
  \begin{align*}      
    \ovl{x}\Big(c\mapsto \exp\big(c r\tau \cdot \tau_2\big)_{ij}\Big)= \lim_\alpha \exp\big(c_\alpha r\tau  \cdot \tau_2\big)_{ij}. 
  \end{align*}
  Since each constant net represents an element of $\RB$ as well, we see that $\pi_\delta|_{\RB}\colon \RB\rightarrow d\cdot  \TO_{\dot\delta(0)}$ is surjective. 
  
For the intersection statement, let $x=0$ and $\vec{n}=\vec{e}_3$. Then, 
  it is clear from \eqref{eq:matrixentr} and \eqref{eq:expformel} that $s\in \pi_\delta(\mathbb{R})\cap d\cdot \TO_{\dot\delta(0)}$ iff $\sin(\pc\tau)=0$, hence $\cos(\pc\tau)=\pm 1$, i.e., $s=\pm d$. 
  
  Finally, combining invariance \eqref{eq:InvGenConnRel} of $\Delta(\ovl{x})$ with bijectivity of $\Co{\sigma}$ and 
  \begin{align*}
    \textstyle\exp\left(\frac{\tau}{2}\cdot\murs(\vec{n})\right)=\alpha(\sigma)\left(\exp\left(\frac{\tau}{2}\cdot \murs(\vec{e}_1)\right)\right),
  \end{align*}    
  we see that $\pi_\delta(\mathbb{R})\cap d\cdot \TO_{\dot\delta(0)}=\{-d,d\}$ holds iff it holds for $\vec{n}=\vec{e}_3$ and $x=0$. 
\item[\ref{lemma:BildCirc1})] 
  By \ref{lemma:BildCirc4}), we have $\mu_1(\pi_\delta(\raisebox{-0.3ex}{$\qR$}))= \mu_1(\pi_\delta(\RR))+\mu_1\big(\TO_{\dot\delta(0)}\big)$ for $\mu_1$ the Haar measure on $\SU$.
  Now, the derivatives of the functions \eqref{eq:cappluco} and \eqref{eq:cappluco2} are given by
  \begin{align*}
  \dot a(c)&= -\textstyle\frac{cr^2\tau}{\pc}\sin(\pc\tau)+ \textstyle\frac{\I c r^2}{2\pc^2}\left[\cos(\pc \tau)\tau-\textstyle\frac{1}{\pc}\sin(\pc\tau)\right]\\
  \dot b(c)&=\textstyle\frac{r}{\pc} \sin(\pc \tau) - \textstyle\frac{c^2r^3}{\pc^3} \sin(\pc \tau) +     \textstyle\frac{c^2r^3\tau}{\pc^2} \cos(\pc \tau),
  \end{align*} 
  so that we have $\dot b(0)= 2 r \sin\!\big(\textstyle\frac{\tau}{2}\big)\neq 0$ because $0<\tau<2\pi$ and $\dot a(c)\neq 0$ if $c\neq 0$. 
  Consequently, 
  $\pi_{\delta}(\RR)$ can be decomposed into countably many $1$-dimensional embedded submanifolds, each of measure zero w.r.t.\ the Haar measure on $\SU$. Thus, $\mu_1(\pi_{\delta}(\RR))=0$, and obviously we have $\mu_1\big(\TO_{\dot{\delta}(0)}\big)=0$ as well.   
\item[\ref{lemma:BildCirc3})]
  By invariance \eqref{eq:InvGenConnRel} of $\Delta(\ovl{x})$ and bijectivity of $\alpha_\sigma$ for $\sigma\in\SU$, it suffices to consider the case where $\delta=\gamma_{\tau,r}$ holds. Then, for $\x\in \qR$, we have $\x\in \RB$ iff
  \begin{align}
  \label{eq:dhdhd}
  \left(\pi_\delta(\ovl{x})\right)_{11}-\e^{-\frac{\I}{2}\tau}\cdot (\pi_{\gamma_{\tau r}}(\ovl{x}))_{11}=0.
  \end{align}
  In fact, for $\delta$ as above, we have $d=\exp(\frac{\tau}{2}\cdot \tau_3)$, so that we can choose $\sigma=\me$. Then, we see from the first equation in the proof of \ref{lemma:BildCirc4})that  
  \begin{align*}
  	\left(\pi_\delta(\ovl{x})\right)_{11}=\e^{-\frac{\I}{2}\tau}\cdot \xi(\x)(a),
  \end{align*}
  and by the same arguments it follows from \eqref{eq:patralinear} that 
  \begin{align}
  \label{eq:hsdfudfsu}
  (\pi_{\gamma_{\tau r}}(\ovl{x}))_{11}= \xi(\x)(\cos(c\cdot r\tau))
  \end{align} 
  holds. Thus, by \eqref{eq:cappluco}, \eqref{eq:dhdhd} equals $\xi(\x)$ applied to the $C_0(\RR)$-part $a_0$ of $\e^{-\frac{\I}{2}\tau}\cdot a$, which gives zero if $\x\in \RB\subseteq \qR$, and something non-zero if $\x\in \RR\subseteq \qR$ because $a_0$ vanishes nowhere.
   
   Now, by \eqref{eq:hsdfudfsu}, each $\ovl{x}\in \RB\subseteq \qR$ is uniquely determined by the values $\{(\pi_{\gamma_l}(\ovl{x}))_{11}\}_{l\in \mathbb{R}_{>0}}$, and the same is true for $\ovl{x}\in \mathbb{R}$ as the functions $\{\chi_l\}_{l\in \mathbb{R}_{>0}}$ separate the points in $\mathbb{R}$.      
\item[\ref{lemma:BildCirc5})]

  Again, it suffices to show the claim for $\delta=\gamma_{\tau,r}$. 
  
  For this, let $x\neq y$ with $\pi_{\gamma_{\tau,r}}(x)=\pi_{\gamma_{\tau,r}}(y)$. If $x=-y$, a closer look at the entry $(\pi_{\gamma_{\tau,r}})_{12}$ (the function $b$) shows that $\sin(\beta_x \tau)=0=\sin(\beta_y\tau)$ holds, hence $x=a_n$ and $y=a_{-n}$ for some $n\in \ZZ_{\neq 0}$. 
If $|x|\neq |y|$, then $\beta_x\neq \beta_y$, so that a closer look at
the entry $(\pi_{\tau,r})_{11}$ (function $a$) shows\footnote{Consider the graph of the curve $\RR_{\geq 0}\ni \beta\mapsto\cos(\beta \tau)\vec{e}_1+ \frac{1}{2\beta}\sin(\beta \tau)\vec{e}_2\in \RR^2$. Then, compare its self intersection points with the zeroes of the function  
	$b\colon c\mapsto \frac{cr}{\beta_c}\sin(\beta_c \tau)$.} 
that  either 
$\tau\beta_x, \tau\beta_y \in \{2\pi n\:|\: n\in \mathbb{N}_{\geq 1}\}$ or $\tau\beta_x, \tau \beta_y\in \{(2n-1)\pi\:|\: n\in \mathbb{N}_{\geq 1}\}$ holds. 
Thus, $x=a_n$ and $y=a_m$ for some $m,n\in \ZZ_{\neq 0}$, 
from which    
the first part follows. 
The merging property is immediate from the formulas \eqref{eq:cappluco} and \eqref{eq:cappluco2}. 
\hspace*{\fill}{\footnotesize$\blacksquare$}
\end{enumerate} 

\section{Projective Limits}
In this section, we will adapt the standard facts on projective structures to our definitions from Subsection \ref{subsec:ProjStruc}.
\begin{lemma} 
  \label{lemma:equivalence}
  Let $X$ be a projective limit of $\{X_\alpha\}_{\alpha \in I}$ w.r.t.\ the maps $\pi_\alpha \colon X\rightarrow X_\alpha$ for $\alpha\in I$ and $\pi^{\alpha_2}_{\alpha_1}\colon X_{\alpha_2}\rightarrow X_{\alpha_1}$ for $\alpha_1,\alpha_2 \in I$ with $\alpha_2 \geq \alpha_1$. Then, $X$ is homeomorphic to 
  \begin{align} 
  \label{eq:projlimprod} 
    \widehat{X}:=\left\{\hx \in \prod_{\alpha\in I}X_\alpha\: \:\Bigg|\:\: \pi_{\alpha_1}^{\alpha_2}(x_{\alpha_2})=x_{\alpha_1}\:\: \forall\:\alpha_2\geq \alpha_1\right\}
  \end{align}  
  equipped with the subspace topology inherited from the Tychonoff topology on $\prod_{\alpha\in I}X_\alpha$. 
  \begin{proof}
    The map $\eta\colon X\rightarrow \widehat{X}$,\:\: $x\mapsto \{\pi_\alpha(x)\}_{\alpha\in I}$ is well defined by Definition \ref{def:ProjLim}.\ref{def:ProjLim2}. Moreover, $\eta$ is continuous as it is so as a map from $X$ to $\prod_{\alpha\in I}X_\alpha$. This is  because $\pr_\alpha \cp \eta=\pi_\alpha$ is continuous for all projection maps $\pr_\alpha\colon\prod_{\alpha\in I}X_\alpha \rightarrow X_\alpha$ just by assumption, see Definition \ref{def:ProjLim}.\ref{def:ProjLim1}. 
    Thus, $\eta$ is a homeomorphism if it is bijective because $X$ is compact and $\widehat{X}$ is Hausdorff. 
    Now, injectivity of $\eta$ is immediate from Definition \ref{def:ProjLim}.\ref{def:ProjLim3}. For surjectivity, assume that $\{x_\alpha\}_{\alpha\in I}=\hx\in \widehat{X}$ with $\hx\notin \im[\eta]$, i.e.,
    $\bigcap_{\alpha\in I}\pi_\alpha^{-1}(x_\alpha)=\eta^{-1}(\hx)=\emptyset$. By continuity of $\pi_\alpha$ the sets $\pi_\alpha^{-1}(x_\alpha)\subseteq X$ are closed, hence compact by compactness of $X$. Consequently, there are finitely many $\alpha_1,{\dots},\alpha_k \in I$ such that $\pi_{\alpha_1}^{-1}(x_{\alpha_1})\cap {\dots} \cap \pi_{\alpha_k}^{-1}(x_{\alpha_k})=\emptyset$. By directedness of $I$ we find some $\alpha\in I$ with $\alpha_j\leq \alpha$ for all $1\leq j\leq k$, hence
    \begin{align}
      \label{eq:projlimmm}    
      \{x_{\alpha_j}\}= \big(\pi_{\alpha_j}^{\alpha}\cp \pi_\alpha\big)\left(\pi_\alpha^{-1}(x_\alpha)\right)= \pi_{\alpha_j}\left(\pi_\alpha^{-1}(x_\alpha)\right)\quad  \text{for all}\quad 1\leq j\leq k
    \end{align}
    because $\pi_\alpha^{-1}(x_\alpha)$ is non-empty ($\pi_\alpha$ is surjective) and $\pi^{\alpha}_{\alpha_j}(x_\alpha)=x_{\alpha_j}$ for all $1\leq j\leq k$. Applying $\pi_{\alpha_j}^{-1}$ to both sides of \eqref{eq:projlimmm} gives $\pi_{\alpha_j}^{-1}(x_{\alpha_j})\supseteq\pi_\alpha^{-1}(x_\alpha)$ for all $1\leq j\leq k$, which contradicts that $\pi_{\alpha_1}^{-1}(x_{\alpha_1})\cap {\dots} \cap \pi_{\alpha_k}^{-1}(x_{\alpha_k})=\emptyset$ holds.  
  \end{proof}
\end{lemma}

\begin{lemma}
  \label{lemma:ConstMeas}
  Let $X$ and $\{X_\alpha\}_{\alpha\in I}$ be as in Definition \ref{def:ProjLim}. Then, the normalized Radon measures on $X$ are in bijection with the consistent families of normalized Radon measures on $\{X_\alpha\}_{\alpha\in I}$.
\end{lemma}
\begin{proof}
  For $\mu$ a normalized Radon measure on $X$, it is straightforward to see that $\{\pi_\alpha(\mu)\}_{\alpha\in I}$ is a consistent family of normalized Radon measures on $\{X_\alpha\}_{\alpha\in I}$. For the converse statement, define $\mathrm{Cyl}(X):=\bigcup_{\alpha\in I}\pi_\alpha^*(C(X_\alpha))\subseteq C(X)$. Then, $\mathrm{Cyl}(X)$ is closed under involution, separates the points in $X$, and vanishes nowhere. Moreover, $\mathrm{Cyl}(X)$ is closed under addition because 
    \begin{align}
      \label{eq:addd}
      f\cp \pi_{\alpha_1} + g\cp \pi_{\alpha_2}&=\left(f \cp\pi^{\alpha_3}_{\alpha_1}\right) \cp \pi_{\alpha_3}+ \left(g \cp\pi^{\alpha_3}_{\alpha_2}\right)\cp \pi_{\alpha_3}\in \mathrm{Cyl}(X),
    \end{align}
    where $\alpha_1,\alpha_2,\alpha_3\in I$ with $\alpha_1,\alpha_2\leq \alpha_3$.
    It follows in the same way that $\mathrm{Cyl}(X)$ is closed under multiplication.
    By the Stone-Weierstrass theorem, $\mathrm{Cyl}(X)$ is a dense $^*$-subalgebra of $C(X)$. Moreover, the map
    \begin{align*}
    \III\colon \mathrm{Cyl}(X)\rightarrow \mathbb{C},\quad f\cp \pi_\alpha \mapsto \int_{X_\alpha}f\:\dd\mu_\alpha
    \end{align*}
    is well defined, linear, and continuous w.r.t.\ the uniform norm on $C(X)$. In fact, if $f\cp \pi_{\alpha_1}= g\cp \pi_{\alpha_2}$, then 
\begin{align*}
	   \big(f \cp\pi^{\alpha_3}_{\alpha_1}\big) \cp \pi_{\alpha_3}= \big(g \cp\pi^{\alpha_3}_{\alpha_2}\big)\cp \pi_{\alpha_3}\qquad\text{if}\qquad\alpha_1,\alpha_2\leq \alpha_3.
\end{align*}    
  Hence, $f \cp\pi^{\alpha_3}_{\alpha_1}=g \cp\pi^{\alpha_3}_{\alpha_2}$ by surjectivity of $\pi_{\alpha_3}$ so that the transformation formula yields
    \begin{align*}
      \III(f\cp \pi_{\alpha_1})&=\int_{X_{\alpha_1}}f\: \dd\mu_{\alpha_1} = \int_{X_{\alpha_3}}f\cp \pi_{\alpha_1}^{\alpha_3}\: \dd\mu_{\alpha_3}\\
      &= \int_{X_{\alpha_3}}g\cp \pi_{\alpha_2}^{\alpha_3}\: \dd\mu_{\alpha_3}=\int_{X_{\alpha_2}}g\: \dd\mu_{\alpha_2}=\III(g\cp \pi_{\alpha_2}).
    \end{align*}
    This shows well-definedness of $\III$. Linearity is clear from \eqref{eq:addd}, and continuity follows from surjectivity of $\pi_\alpha$ by
	\begin{align*}
		|\hs\III(f\cp \pi_\alpha)\hs|\leq\|f\|_\infty=\|f\cp \pi_\alpha \|_\infty.
	\end{align*}	    
   	Since $\III$ is linear and continuous, it extends to a continuous linear functional on $C(X)$ which, by the Riesz-Markov theorem (see, e.g., 2.5 Satz in \cite[Chap.VIII, \S 2]{Elstrodt}), defines a finite Radon measure $\mu$ on $\Borel(X)$. Then, 
    $\mu(X)=\III(1)=1$ and for each $f\in C(X_\alpha)$ we have
    \begin{align*}
      \int_{X_\alpha}f\:\dd(\pi_\alpha^*\mu)=\int_X (f\cp \pi_\alpha) \:\dd\mu=\III(f\cp \pi_\alpha)=\int_{X_\alpha}f\:\dd\mu_\alpha,
    \end{align*}
    hence $\pi_\alpha^*\mu=\mu_\alpha$ again by the Riesz-Markov theorem. Finally, if $\mu'$ is a further finite Radon measure with $\pi_\alpha^*\mu'=\mu_\alpha$ for all $\alpha\in I$, then $\III'\colon C(X)\rightarrow \mathbb{C}$, $f\mapsto \int_{X}f\:\dd\mu'$ is continuous and we have $\III|_{\mathrm{Cyl}(X)}=\III'|_{\mathrm{Cyl}(X)}$ by the transformation formula. Thus, $\III=\III'$ by denseness of $\mathrm{Cyl}(X)$ in $C(X)$, hence $\mu=\mu'$.
\end{proof}

\section{The Ashtekar-Lewandowski Measure}
\label{sec:AshLewMeasure}
In this subsection, we will reformulate the results from \cite{{ProjTechAL}} in terms of the Definitions \ref{def:ProjLim} and \ref{def:Mass}. 
For this, let $P=\RR^3\times \SU$, $\Paths$ the set of embedded analytic curves in $\RR^3$ and $\Gamma:=\bigsqcup_{n=1}^\infty \Paths^n$. To simplify the notations, we will write $S$ instead of $\SU$ in the sequel. 
\begingroup
\setlength{\leftmargini}{14pt}
\begin{itemize}
\item
  \itspacecc
  A refinement of $(\gamma_1,{\dots},\gamma_k)\in\Gamma$ is an element $(\delta_1,{\dots},\delta_n)\in \Gamma$ such that for each $\gamma_j$ we find a decomposition $\{(\gamma_j)_i\}_{1\leq i\leq k_j}$ such that each restriction $(\gamma_j)_i$ is equivalent to one of the curves $\delta_r$ or $\delta_r^{-1}$ for some $1\leq r\leq n$.
\item
  \itspace
  An element $(\delta_1,{\dots},\delta_n)\in \Gamma$ is said to be independent iff for each collection $\{s_1,{\dots},s_n\}\subseteq S$  
  there is some $\w\in \Con$ with $h_{\w}(\delta_i)= s_i$ for all $1\leq i\leq n$. 
\item
  \itspace
  Let $\gr\subseteq \Gamma$ be the set of finite tuples $(\gamma_1,{\dots},\gamma_k)\in \Gamma$ with 
  \begin{align*}
    \im[\gamma_i]\cap \im[\gamma_j]\subseteq \{\gamma_i(a_i),\gamma_i(b_i),\gamma_j(a_j),\gamma_j(b_j)\}\quad\text{ for }\quad 1\leq i\neq j\leq k
  \end{align*}
  where $\dom[\gamma_i]=[a_i,b_i]$.
\item
  \itspace
  For $(\gamma_1,{\dots},\gamma_k), (\gamma'_1,{\dots},\gamma'_l)\in \gr$ write $(\gamma_1,{\dots},\gamma_k)\leq (\gamma'_1,{\dots},\gamma'_l)$ 
  iff each $\gamma_i$ admits a decomposition $\{(\gamma_i)_j\}_{1\leq j\leq s_i}$ such that each restriction $(\gamma_j)_i$ is equivalent to one of the paths $\gamma'_1,{\dots},\gamma'_l$ or its inverse.
\end{itemize}
\endgroup
\noindent
Each $(\gamma_1,{\dots},\gamma_k)\in \gr$ is independent\footnote{Cf. Section 3 in \cite{Ashtekar2008} or Proposition A.1 in \cite{ParallTranspInWebs}.}
and $(\gr,\leq)$ is directed. In fact, if $(\gamma_1,{\dots},\gamma_k), (\gamma'_1,{\dots},\gamma'_l)$ $\in \gr$, the proof of Lemma D.1.3 in \cite{Thesis} shows that for $(\gamma_1,{\dots},\gamma_k, \gamma'_1,{\dots},\gamma'_l)\in \Gamma$ we find a refinement $(\delta_1,{\dots},\delta_n)\in \Gamma$, such that $\im[\delta_i]\cap \im[\delta_j]$ is finite for all $1\leq i\neq j\leq n$. Thus, splitting $(\delta_1,{\dots},\delta_n)$ at the respective intersection points gives the desired upper bound in $\gr$.
\begin{definition}
  \label{def:ProjLimAL}
  \begin{enumerate}
  \item
    For $\alpha=(\gamma_1,{\dots},\gamma_k)\in \gr$, $|\alpha|:=k$ we define the map $\pi_\alpha\colon \ovl{\Con}\rightarrow X_\alpha:=S^{|\alpha|}$ by
    \begin{align*}	
      \pi_\alpha(\ovl{\w}):= \left(\homiso(\ovl{\w})(\gamma_1),{\dots},\homiso(\ovl{\w})(\gamma_k)\right).
    \end{align*}
    This map is surjective by independence of $(\gamma_1,{\dots},\gamma_k)\in\gr$ and because $\Con$ is embedded into $\A$.
  \item
    Let $\alpha=(\gamma_1,{\dots},\gamma_k)\leq (\gamma'_1,{\dots},\gamma'_{k'})=\alpha'$ and $\{(\gamma_i)_j\}_{1\leq j\leq l_i}$ the corresponding decomposition of $\gamma_i$ for $1\leq i\leq k$.
    Then,
	\begin{align*}
	(\gamma_i)_j^{p_{ij}}\sim_\Con \gamma'_{m_{ij}} 
	\end{align*}
	for $p_{ij}\in\{1,-1\}$ and $1\leq m_{ij} \leq k'$ 
    uniquely determined by independence of $(\gamma'_1,{\dots},\gamma'_{k'})$.
    
    We define $\pi^{\alpha'}_\alpha\colon X_{\alpha'}\rightarrow X_\alpha$ by
    \begin{align}
      \label{eq:alphaalphsastrich}
      \pi^{\alpha'}_\alpha\left(x_1,{\dots},x_{k'}\right):=\Bigg(\prod_{j=1}^{l_1}\left(x_{m_{1j}}\right)^{p_{1j}},{\dots},\prod_{j=1}^{l_k}\left(x_{m_{kj}}\right)^{p_{kj}}\Bigg)
    \end{align}
    which is obviously continuous. Observe that \eqref{eq:alphaalphsastrich} cannot not depend on the decompositions $\{(\gamma_i)_j\}_{1\leq j\leq l_i}$ because $\pi_{\alpha'}$ is surjective and 
   $\pi^{\alpha'}_\alpha\cp \pi_{\alpha'}=\pi_\alpha$ holds. 
  \end{enumerate}
\end{definition}
\begin{lemma}
  \label{lemma:consFam}
  \begin{enumerate}
  \item
   \label{lemma:consFam1}
    $\ovl{\Con}$ is a projective limit of $\{X_\alpha\}_{\alpha\in \gr}$ w.r.t.\ the maps $\pi_\alpha$ for $\alpha \in \gr$ and $\pi_\alpha^{\alpha'}$ for $\alpha,\alpha'\in \gr$ with $\alpha \leq \alpha'$. 
  \item
   \label{lemma:consFam2}
    Let $\mu_k$ denote the Haar measure on $S^{k}$ for $k\in \mathbb{N}_{\geq 1}$ and $\mu_\alpha:=\mu_{|\alpha|}$ for $\alpha\in \gr$. Then, $\{\mu_{\alpha}\}_{\alpha\in \gr}$ is a consistent family of normalized Radon measures w.r.t.\ $\{X_\alpha\}_{\alpha\in \gr}$.
  \end{enumerate}
  \begin{proof}
    \begin{enumerate}
    \item
      For continuity of the maps $\pi_\alpha$ it suffices to consider the case where $\alpha=\gamma \in \Paths$ holds. 
      But, then continuity is clear from \eqref{eq:pigammawquer} because $[h_\gamma]_{ij}\in \cC$ and $\ovl{\w}\in \Spec(\cC)$.
      For the separation property \ref{def:ProjLim3}) from Definition \ref{def:ProjLim} observe that the Gelfand transforms of the functions $[h_\gamma]_{ij}$ separate the points in $\ovl{\Con}$ as they generate the continuous functions on $\ovl{\Con}$. Thus, the claim is clear from \eqref{eq:pigammawquer} as well.
    \item
      We have to show that $\pi^{{\alpha'}}_{\alpha}(\mu_{{\alpha'}})=\mu_{\alpha}$ holds if $\alpha\leq {\alpha'}$. By Riesz-Markov, for this 
      it suffices to verify $\int_{X_\alpha}f \:\dd\mu_\alpha=\int_{X_\alpha}f \:\dd\pi^{{\alpha'}}_{\alpha}(\mu_{{\alpha'}})$ for all $f\in C(X_\alpha)$. Now,  
      \begin{align*}
        \Borel(X_{\alpha'})=\bigotimes_{1\leq i\leq k'=|\alpha'|}\Borel(S)
      \end{align*}
      because $S$ is second countable so that, since $S$, $S^{k'-1}$ are $\sigma$-finite, for $g\in C(X_{\alpha'})$, $x=(x_1,{\dots},x_{k'})\in S^{k'}$ and $1\leq i\leq k'$  Fubini's formula gives
      \begin{align}
        \label{eq:Fubini}
        \int_{X_{\alpha'}}g \:\dd\mu_{\alpha'}=  \int_{S^{k'-1}} \left(\int_S g(x) \:\dd\mu_1(x_i)\right)\dd\mu_{k'-1}(\raisebox{-0.2ex}{$x^i$})
      \end{align}
      for $S^{k'-1}\ni x^i:=(x_1,{\dots},x_{i-1},x_{i+1},{\dots},x_{k'})$. Hence, for $f\in C(X_\alpha)$ we have
      \begin{equation*}
        \begin{split}
          \int_{X_\alpha}f \:\dd\pi^{{\alpha'}}_{\alpha}(\mu_{{\alpha'}})&=\int_{X_{\alpha'}}\Big(f\cp \pi^{{\alpha'}}_{\alpha}\Big) \:\dd\mu_{{\alpha'}}\\ &=\int_{S^{k'-1}}\left(\int_S\Big(f\cp \pi^{{\alpha'}}_{\alpha}\Big)(x) \:\dd\mu_1(x_i)\right)\dd\mu_{k'-1}(\raisebox{-0.2ex}{$x^i$}).
        \end{split}
      \end{equation*}
		By the definition of $\gr$, each of the variables $x_1,{\dots},x_{k'}$ occurs in exactly one of the products on the right hand side of \eqref{eq:alphaalphsastrich}. Consequently, the components $\big[\raisebox{-0.2ex}{$\pi^{{\alpha'}}_{\alpha}$}\big]_i$ of $\pi^{\alpha'}_\alpha$ mutually depend on different variables $\ovl{x}_i=\big(x_{m_{i,1}},{\dots},x_{m_{i,n(i)}}\big)$.
       
       Then, by the left-, right- and inversion invariance of $\mu_1$, for $x\in S^{k'}$ we have 
      \begin{align*}
        \int_S\Big(f\cp \pi^{{\alpha'}}_{\alpha}\Big)(x) \:\dd\mu_1(x_{m_{k,1}})
        & =\int_S f\Big(\big[\raisebox{-0.2ex}{$\pi^{{\alpha'}}_{\alpha}$}\big]_1(\raisebox{-0.1ex}{$\ovl{x}_1$}),{\dots},\big[\raisebox{-0.2ex}{$\pi^{{\alpha'}}_{\alpha}$}\big]_{k-1}(\raisebox{-0.1ex}{$\ovl{x}_{k-1}$}),\big[\raisebox{-0.2ex}{$\pi^{{\alpha'}}_{\alpha}$}\big]_{k}(\raisebox{-0.1ex}{$\ovl{x}_{k}$})\Big) \:\dd\mu_1(x_{m_{k,1}})
        \\
        & =\int_S f\Big(\big[\raisebox{-0.2ex}{$\pi^{{\alpha'}}_{\alpha}$}\big]_1(\raisebox{-0.1ex}{$\ovl{x}_1$}),{\dots},\big[\raisebox{-0.2ex}{$\pi^{{\alpha'}}_{\alpha}$}\big]_{k-1}(\raisebox{-0.1ex}{$\ovl{x}_{k-1}$}),x_{m_{k,1}}\Big) \:\dd\mu_1(x_{m_{k,1}}).
      \end{align*}
      Thus, applying the same argument, inductively 
       we obtain
      \begin{equation*}
        \int_{X_\alpha}f \:\dd\pi^{{\alpha'}}_{\alpha}(\mu_{{\alpha'}}) =\int_{S^{k'}} f(x_{m_{1,1}},{\dots},x_{m_{k,1}})\:\dd\mu_{\alpha'}(x)=\int_{X_\alpha} f\:\dd\mu_{\alpha},
      \end{equation*}
      whereby the last step follows inductively from \eqref{eq:Fubini} and $\mu_1(S)=1$. 
 \end{enumerate}
  \end{proof}
\end{lemma}

\begin{definition}
  \label{def:AshLew}
  The normalized Radon measure $\mAL$ on $\ovl{\Con}$ that corresponds to the consistent family of normalized Radon measures $\{\mu_\alpha\}_{\alpha\in \gr}$ from Lemma \ref{lemma:consFam}.\ref{lemma:consFam2} is called Ashtekar-Lewandowski measure on $\ovl{\Con}$.
\end{definition}

\end{document}